\pdfoutput=1
\RequirePackage{ifpdf}
\ifpdf 
\documentclass[pdftex]{sigma}
\else
\documentclass{sigma}
\fi

\usepackage[mathscr]{eucal}

\numberwithin{equation}{section}

\newtheorem{Theorem}{Theorem}[section]
\newtheorem{Lemma}[Theorem]{Lemma}
\newtheorem{Proposition}[Theorem]{Proposition}
\newtheorem{Conjecture}[Theorem]{Conjecture}
 { \theoremstyle{definition}
\newtheorem{Example}[Theorem]{Example}
\newtheorem{Remark}[Theorem]{Remark} }

\newcommand{\bk}{{\bf k}}
\newcommand{\bb}{{\bf b}}
\newcommand{\bc}{{\bf c}}

\newcommand{\ot}{\otimes}
\newcommand{\ep}{\epsilon}
\newcommand{\U}{\mathcal{U}}
\newcommand{\Rm}{\mathscr{R}}
\newcommand{\Z}{{\mathbb Z}}
\newcommand{\R}{{\mathbb R}}
\newcommand{\C}{{\mathbb C}}
\newcommand{\Q}{{\mathbb Q}}

\begin{document}


\newcommand{\arXivNumber}{1701.07279}

\renewcommand{\thefootnote}{}

\renewcommand{\PaperNumber}{044}

\FirstPageHeading

\ShortArticleName{Integrable Structure of Multispecies Zero Range Process}

\ArticleName{Integrable Structure\\ of Multispecies Zero Range Process\footnote{This paper is a~contribution to the Special Issue on Recent Advances in Quantum Integrable Systems. The full collection is available at \href{http://www.emis.de/journals/SIGMA/RAQIS2016.html}{http://www.emis.de/journals/SIGMA/RAQIS2016.html}}}

\Author{Atsuo KUNIBA~$^\dag$, Masato OKADO~$^\ddag$ and Satoshi WATANABE~$^\dag$}

\AuthorNameForHeading{A.~Kuniba, M.~Okado and S.~Watanabe}

\Address{$^\dag$~Institute of Physics, Graduate School of Arts and Sciences, University of Tokyo,\\
\hphantom{$^\dag$}~Komaba, Tokyo 153-8902, Japan}
\EmailD{\href{mailto:atsuo.s.kuniba@gmail.com}{atsuo.s.kuniba@gmail.com}, \href{mailto:watanabe@gokutan.c.u-tokyo.ac.jp}{watanabe@gokutan.c.u-tokyo.ac.jp}}

\Address{$^\ddag$~Department of Mathematics, Osaka City University,\\
\hphantom{$^\ddag$}~3-3-138, Sugimoto, Sumiyoshi-ku, Osaka, 558-8585, Japan}
\EmailD{\href{mailto:okado@sci.osaka-cu.ac.jp}{okado@sci.osaka-cu.ac.jp}}

\ArticleDates{Received January 26, 2017, in f\/inal form June 07, 2017; Published online June 17, 2017}

\Abstract{We present a brief review on integrability of multispecies zero range process in one dimension introduced recently. The topics range over
stochastic $R$ matrices of quantum af\/f\/ine algebra $U_q\big(A^{(1)}_n\big)$, matrix product construction of stationary states for periodic systems, $q$-boson representation of Zamolodchikov--Faddeev algebra, etc. We also introduce new commuting Markov transfer matrices having a mixed boundary condition and prove the factorization of a family of $R$ matrices associated with the tetrahedron equation and generalized quantum groups at a~special point of the spectral parameter.}


\Classification{81R50; 60C99}

\renewcommand{\thefootnote}{\arabic{footnote}}
\setcounter{footnote}{0}

\section{Introduction}\label{sec1}
Zero range processes (ZRPs) \cite{S} model a variety of stochastic dynamics in biology, chemistry, networks, physics, traf\/f\/ic f\/lows and so forth. Their rich behaviors like condensation, current f\/luctuations and hydrodynamic limit have been important issues in non-equilibrium physics. See for example \cite{EH, GSS, KL} and references therein.

This paper is a brief summary of the integrable multispecies ZRP in one dimension introduced and studied in the recent works \cite{KMMO,KO1,KO2}. We formulate the ZRPs via commuting Markov transfer matrices and present a matrix product formula for stationary probabilities in the periodic boundary condition. The key ingredients in these results are the {\em stochastic $R$ matrix} and the {\em Zamolodchikov--Faddeev $($ZF$)$ algebra}. The subject lies in the intersection of quantum integrable systems and non-equilibrium statistical mechanics. As the title of the paper suggests, we will mainly focus on the former aspect, although we believe the results are essential for
analyzing the physics of the model as far as the stationary properties are concerned.

Quantum $R$ matrices are solutions of the Yang--Baxter equation (YBE) \cite{Bax} and play a most fundamental role in quantum integrable systems~\cite{JB}. They can be systematically produced from the representation theory of quantum groups. It remains, however, a nontrivial problem if an~$R$ matrix can be made {\em stochastic}, namely whether it can be modif\/ied so as to match the basic criteria of Markov matrices which are non-negativity and total probability conservation.

Our stochastic $R$ matrices \cite{KMMO} fulf\/ill the criteria. They originate in the quantum $R$ matrix of the Drinfeld--Jimbo quantum af\/f\/ine algebra $U_q\big(A^{(1)}_n\big)$ in the symmetric tensor representation of general degree. Plainly, they are of type $A$ with arbitrary rank and spin, covering many examples that had been known earlier. Being higher in rank and being analytically continued in spin, it leads to systems with many kinds of particles allowing arbitrarily multiple occupancy at each lattice site, which are characteristic to multispecies ZRPs. These features are reviewed in Sections~\ref{sec2} and~\ref{sec3} based on~\cite{KMMO}. Sections~\ref{sec2.3} and~\ref{sec3.2} also include ZRPs with a new mixed boundary condition.

In Sections \ref{sec4} and \ref{sec5} we turn to the stationary probabilities ${\mathbb P}(\sigma_1,\ldots, \sigma_L)$ of a given conf\/i\-gu\-ra\-tion $(\sigma_1,\ldots, \sigma_L) \in (\Z^n_{\ge 0})^L$ in the ZRPs with the periodic boundary condition. Here~$n$ and~$L$ are the numbers of the species of particles and lattice sites, respectively. We seek the matrix product formula
\begin{gather*}
{\mathbb P}(\sigma_1,\ldots, \sigma_L) = \operatorname{Tr}(X_{\sigma_1}(\mu_1)\cdots X_{\sigma_L}(\mu_L))
\end{gather*}
in terms of a collection of operators $X(\mu) = (X_\alpha(\mu))_{\alpha \in \Z^n_{\ge 0}}$ that satisfy the ZF algebra
\begin{gather*}
X(\mu) \otimes X(\lambda) = \check{\mathscr{S}}(\lambda, \mu)\bigl[X(\lambda) \otimes X(\mu)\bigr].
\end{gather*}
It contains the stochastic $R$ matrix $\check{\mathscr{S}}(\lambda, \mu)$ as the structure function. Here $\lambda, \mu$ can be understood as generic parameters as long as algebraic aspects are concerned, but they are restricted to real numbers in a certain range in the application to the ZRP.

The ZF algebra, originally introduced in the factorized scattering theories in $(1+1)$ dimension \cite{F, ZZ}, has penetrated into the matrix product method in integrable Markov processes in various guises since the 90's. See general remarks in Section~\ref{sec5.1} and also \cite{AL, BE,CDW, CRV,SW1}.

We will review a $q$-boson representation of the ZF algebra obtained in \cite{KO1,KO2}. The simplest nontrivial case is $n=2$ for which it is
\begin{gather*}
X_{\alpha_1, \alpha_2}(\mu) =\frac{\mu^{-\alpha_1-\alpha_2}(\mu;q)_{\alpha_1+\alpha_2}} {(q;q)_{\alpha_1}(q;q)_{\alpha_2}} \frac{(\bb;q)_\infty}{(\mu^{-1}\bb;q)_\infty}\bk^{\alpha_2} \bc^{\alpha_1},
\end{gather*}
where $\bb$, $\bc$, $\bk$ are $q$-boson creation, annihilation and number operators on the Fock space and $(z; q)_m = \prod\limits_{i=0}^{m-1}(1-zq^i)$ is the $q$-shifted factorial. For general $n$, the matrix product operator $X_\alpha(\mu)$ acts on the tensor product of $\frac{1}{2}n(n-1)$ Fock spaces. There are numerous matrix product formulas in terms of bosons known in the literature, most typically for the exclusion processes. See \cite{AL, BE, CRV, DEHP, KMO0, PEM} for example and references therein. Our result (Theorem~\ref{th:nzm}) is the f\/irst example distinct from them involving a quantum dilogarithm type {\em infinite product} of $q$-bosons.

One of the key facts in our approach is the explicit factorized formula (\ref{fac}) of an $R$ matrix of $U_q\big(A^{(1)}_n\big)$ at a special point of the spectral parameter. In the last Section~\ref{sec6} we seek a similar result for a family of~$R$ matrices associated with the generalized quantum groups labeled by $(\epsilon_1,\ldots, \epsilon_{n+1}) \in \{0,1\}^{n+1}$. They are constructed from $(n+1)$-fold product of the solutions to the tetrahedron equation~\cite{Zam80}, a three-dimensional~(3D) generalization of the YBE, called 3D $R$ $(\epsilon_i=0)$ and 3D $L$ $(\epsilon_i=1)$~\cite{KOS}. The stochastic $R$ matrix in Sections~\ref{sec2}--\ref{sec5} originates in the special case $\epsilon_1= \cdots =\epsilon_{n+1}=0$. We present the Serre type relations of the relevant generalized quantum groups explicitly and prove a similar factorized formula for $(\epsilon_1,\ldots, \epsilon_{n+1})$ of the form
$(1^\kappa,0^{n+1-\kappa})\, (0 \le \kappa \le n+1)$. These results are new. Their application is yet to be explored.

The layout of the paper is as follows. In Section~\ref{sec2} we recall two kinds of stochastic~$R$ matrices~$S(z)$ and $\mathscr{S}(\lambda,\mu)$, and
construct several kinds of commuting transfer matrices from them. In Section~\ref{sec3} we specialize these transfer matrices to formulate integrable multispecies ZRPs. They include discrete and continuous time models with both periodic and mixed boundary conditions. The latter is new. In Section~\ref{sec4} stationary states of these ZRPs are studied, and its matrix product construction is linked to the ZF algebra for the models with periodic boundary condition.
In Section~\ref{sec5} we make general remarks on ZF algebra and give a $q$-boson representation when the structure function is the stochastic~$R$ matrix $\mathscr{S}(\lambda,\mu)$. It yields the stationary probabilities in the matrix product form for the associated $n$-species ZRP. This part is a review of \cite{KO1, KO2}. In Section~\ref{sec6} we extend the factorization~(\ref{fac}) to the $R$ matrices for a class of generalized quantum groups. The result is presented with some background connected to the tetrahedron equation~\cite{KOS}. Section \ref{sec7} is a short summary. Appendix~\ref{sec.app} contains the explicit form of the quantum $R$ matrix for the generalized quantum group $\U_A(1,1,0)$.

Throughout the paper we use the notation $\Z_n = \Z/n\Z$, $\theta(\mathrm{true})=1$, $\theta(\mathrm{false}) =0$, the $q$-shifted factorial $(z)_m = (z; q)_m = \prod\limits_{j=0}^{m-1}(1-zq^j)$ and the $q$-binomial $\binom{m}{k}_{\!q} = \theta(k \in [0,m]) \frac{(q)_m}{(q)_k(q)_{m-k}}$. The symbols $(z)_m$ appearing in this paper always mean $(z; q)_m$. For integer arrays $\alpha=(\alpha_1,\ldots, \alpha_m), \beta=(\beta_1,\ldots, \beta_m)$ of {\em any} length $m$, we write $|\alpha | = \alpha_1+\cdots + \alpha_m$ and the Kronecker delta as $\delta_{\alpha, \beta} =\delta^{\alpha}_{\beta} = \prod\limits_{i=1}^m\theta(\alpha_i=\beta_i)$. The letter $\delta$ will also be used extensively to mean a~local state and in such circumstances we will use the notation $\theta(\alpha=\beta)$ more frequently than~$\delta_{\alpha, \beta}$ to avoid confusion. The relation $\alpha \le \beta$ or equivalently
$\beta \ge \alpha$ is def\/ined by $\beta-\alpha \in \Z^m_{\ge 0}$. We often denote by $0$ or $0^m$ to mean $(0, \ldots, 0) \in \Z^m_{\ge 0}$.

While preparing the text, we were informed of the paper \cite{Kuan}, where the author obtains Markov duality functions for the models treated in this paper.

\section{Commuting transfer matrices}\label{sec2}
\subsection[Stochastic $R$ matrice]{Stochastic $\boldsymbol{R}$ matrices}\label{sec2.1}
Let us recall the stochastic $R$ matrices $S(z)$ and $\mathscr{S}(\lambda, \mu)$ \cite{KMMO} associated with the Drinfeld--Jimbo quantum af\/f\/ine algebra $U_q\big(A^{(1)}_n\big)$. They are constructed by suitably modifying the quantum $R$ matrix characterized by~(\ref{wsy}).

For $l \in \Z_{\ge 1}$, introduce the vector space $V_l$ whose basis is labeled with the set $B_l$ as
\begin{gather}\label{BV}
B_l = \big\{\alpha=(\alpha_1,\ldots, \alpha_{n+1}) \in \Z_{\ge 0}^{n+1}
\,|\, |\alpha| = l\big\},\qquad V_l = \bigoplus_{(\alpha_1,\ldots, \alpha_{n+1}) \in B_l}\C |\alpha_1,\ldots, \alpha_{n+1}\rangle.
\end{gather}
We write $|\alpha_1,\ldots, \alpha_{n+1}\rangle$ simply as $|\alpha \rangle$. There is an algebra homomorphism $U_q\big(A^{(1)}_n\big) \rightarrow \operatorname{End}(V_l)$ called the symmetric tensor representation of degree~$l$ depending on a spectral parameter. We are concerned with the standard quantum $R$ matrix $R(z)=R^{l,m}(z)$ living in $\operatorname{End}(V_l \otimes V_m)$. Leaving the representation theoretical background to Section~\ref{sec6}, we present an explicit formula:
\begin{gather}
R(z) (|\alpha\rangle \otimes | \beta\rangle ) = \sum_{\gamma \in B_l,\delta \in B_m} R(z)_{\alpha,\beta}^{\gamma,\delta}|\gamma\rangle \otimes | \delta\rangle, \\
R(z)_{\alpha,\beta}^{\gamma,\delta}=\frac{z^{-m}(q^{l-m}z;q^2)_{m+1}}{(q^{l-m+2}z^{-1};q^2)_m}\sum_{c_0, \ldots, c_{n} \in \Z_{\ge 0}}
z^{c_0}\Rm^{\gamma_1, \delta_1, c_0}_{\alpha_1, \beta_1, c_1}\cdots\Rm^{\gamma_{n}, \delta_{n}, c_{n-1}}_{\alpha_{n}, \beta_{n}, c_{n}}
\Rm^{\gamma_{n+1}, \delta_{n+1}, c_{n}}_{\alpha_{n+1}, \beta_{n+1}, c_0},\label{rel}\\
\Rm^{a,b,c}_{i,j,k} = \delta^{a+b}_{i+j}\delta^{b+c}_{j+k} q^{ik+b} \oint\frac{du}{2\pi {\mathrm i}u^{b+1}}\frac{(-q^{2+a+c}u;q^2)_\infty(-q^{-i-k}u;q^2)_\infty}{(-q^{a-c}u;q^2)_\infty(-q^{c-a}u;q^2)_\infty}\in \Z[q]. \label{Rint}
\end{gather}
The integral encircles the origin $u=0$ anti-clockwise to pick the residues. $z$ is called the spectral parameter. Explicit formulas for $\Rm^{a,b,c}_{i,j,k}$ are available for example in \cite[equation~(2.2)]{KOS}. The fact that $\Rm^{a,b,c}_{i,j,k} \in \Z[q]$ can be seen from them. Owing to the factor $\delta^{a+b}_{i+j}\delta^{b+c}_{j+k}$ in~(\ref{Rint}), the sum~(\ref{rel}) consists of f\/initely many terms and $R(z)_{\alpha,\beta}^{\gamma,\delta}$ is a rational function of $z$ and $q$. The prefactor in~(\ref{rel}) has been chosen so as to achieve the normalization~(\ref{yuk}), which will ultimately lead to~(\ref{sum1}) related to the total probability conservation.

The set of $q$-polynomials $\{\Rm^{a,b,c}_{i,j,k}\}$ form a solution of the tetrahedron equation~\cite{Zam80} having an origin in the quantized coordinate ring of ${\rm SL}_3$~\cite{KV}. It was stated that the composition~(\ref{rel}) yields the quantum~$R$ matrix in~\cite{BS} and proved in \cite[Appendix~B]{KO0}. The formula~(\ref{Rint}) is due to~\cite{Ser}. See Section~\ref{sec6} and~\cite{KO0,KOS} for a further explanation and generalization. For recent progress on evaluating the sum (\ref{rel}), we refer to~\cite{BM}.

The f\/irst stochastic $R$ matrix $S(z)=S^{l,m}(z)\in \operatorname{End}(V_l \otimes V_m)$ is obtained just by taking the {\em stochastic gauge} of $R(z)$ as follows:
\begin{gather}
S(z) (|\alpha\rangle \otimes | \beta\rangle ) = \sum_{\gamma \in B_l,\delta \in B_m} S(z)_{\alpha,\beta}^{\gamma,\delta} |\gamma\rangle \otimes | \delta\rangle,\qquad S(z)^{\gamma,\delta}_{\alpha, \beta} = q^\eta R(z)^{\gamma,\delta}_{\alpha, \beta},\nonumber\\
\eta= \sum_{1 \le i<j \le n+1} (\delta_i\gamma_j - \alpha_i\beta_j).\label{SR}
\end{gather}

\begin{Example}\label{ykw}
Consider the simplest example $S(z)=S^{1,1}(z)$. We denote $S(z)^{{\bf e}_i,{\bf e}_j}_{{\bf e}_k,{\bf e}_l}$ simply by~$S(z)^{i,j}_{k,l}$, where ${\bf e}_i$ is the $i$ th basis vector def\/ined in~(\ref{hrm}). By the graphical representation~(\ref{vertex}), nonzero elements are given by
\begin{gather*}
\includegraphics{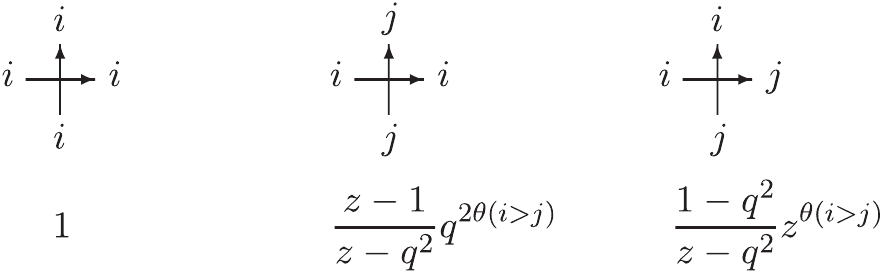}
\end{gather*}
where $1 \le i \neq j \le n+1$. They lead to the $n$-species {\em symmetric simple exclusion process} at $q=1$ and {\em asymmetric simple exclusion process} for $q\ne 1$. See Section~\ref{sec3.3}.
\end{Example}

\begin{Theorem}[\cite{KMMO}]\label{noi}
Set $z_{i,j}=z_i/z_j$. Then the following equalities are valid:
\begin{alignat}{3}
&\text{YBE}\colon \quad &&
S_{1,2}^{k,l}(z_{1,2}) S_{1,3}^{k,m}(z_{1,3})S_{2,3}^{l,m}(z_{2,3})= S_{2,3}^{l,m}(z_{2,3})S_{1,3}^{k,m}(z_{1,3})S_{1,2}^{k,l}(z_{1,2}),&\label{ybe1}\\
&\text{sum-to-unity}\colon \quad && \sum_{\gamma \in B_l, \, \delta \in B_m}S^{l,m}(z)^{\gamma,\delta}_{\alpha, \beta} =1,\qquad
\forall\, (\alpha, \beta) \in B_l \times B_m,&\label{sum1}\\
&\text{factorization}\colon \quad && S^{l,m}\big(z=q^{l-m}\big)^{\gamma,\delta}_{\alpha, \beta}
= \delta_{\alpha+\beta}^{\gamma+\delta}
\Phi_{q^2}\big(\bar{\gamma}| \bar{\beta}; q^{-2l},q^{-2m}\big),\qquad l \le m,&\label{fac}
\end{alignat}
where $\bar{\gamma}=(\gamma_1,\ldots, \gamma_n)$, $\bar{\beta}=(\beta_1,\ldots, \beta_n) \in \Z^n_{\ge 0}$ are the arrays $\gamma$, $\beta$ with the $(n+1)$-th component dropped.
\end{Theorem}

Here $S_{1,2}^{k,l}(z_{1,2})$ for instance denotes the matrix that acts as $S^{k,l}(z_{1,2})$ on the f\/irst and the second components from the left in $V_k \otimes V_l \otimes V_m$. The $S^{l,m}(z)^{\gamma,\delta}_{\alpha, \beta}$ is an element of the mat\-rix~$S^{l,m}(z)$.

In (\ref{ybe1}) and (\ref{sum1}), there is no constraint like $l \le m$ in~(\ref{fac}). The $\Phi_q(\gamma | \beta; \lambda, \mu)$ appearing in~(\ref{fac}) is the function of $n$-component arrays~$\beta$,~$\gamma$ and parameters~$q$, $\lambda$, $\mu$ def\/ined by
\begin{gather}
\Phi_q(\gamma|\beta; \lambda,\mu) =q^{\varphi(\beta-\gamma,\gamma)} \left(\frac{\mu}{\lambda}\right)^{|\gamma|}
\frac{(\lambda;q)_{|\gamma|}(\frac{\mu}{\lambda};q)_{|\beta|-|\gamma|}}{(\mu;q)_{|\beta|}}
\prod_{i=1}^{n}\binom{\beta_i}{\gamma_i}_{q},\nonumber\\
\varphi(\alpha, \beta) = \sum_{1 \le i<j \le n} \alpha_i\beta_j.\label{Pdef}
\end{gather}
By the def\/inition $\Phi_q(\gamma|\beta;\lambda, \mu) =0$ unless $\gamma \le \beta$. The modif\/ication by the factor $q^\eta$ in~(\ref{SR}) does not spoil the YBE. The point is that it can be so chosen that the sum-to-unity property~(\ref{sum1}) holds. It will eventually lead to the total probability conservation in the relevant stochastic models in what follows. The factorization~(\ref{fac}) at the special point $z=q^{l-m}$ is also nontrivial, and assures the non-negativity of the transition rate manifestly in an appropriate range of parameters. We will generalize a formula like~(\ref{fac}) to a wider class of~$R$ matrices in Section~\ref{sec6}.

The second stochastic $R$ matrix $\mathscr{S}(\lambda, \mu)$ is extracted essentially from $(\ref{fac})|_{q^2\rightarrow q}$ by regar\-ding~$q^{-l}$,~$q^{-m}$ as parameters $\lambda$, $\mu$. It is a linear operator on $W \otimes W$ with $W$ def\/ined by $W = \bigoplus_{(\alpha_1,\ldots, \alpha_n) \in \Z_{\ge 0}^n}
\C|\alpha_1, \ldots, \alpha_n \rangle$. The basis $|\alpha_1, \ldots, \alpha_n \rangle$ here is labeled with an $n$-component array as opposed to~(\ref{BV}), but we also denote it by the same symbol $|\alpha\rangle$ for simplicity. Then $\mathscr{S}(\lambda, \mu)$ is def\/ined by
\begin{gather}
\mathscr{S}(\lambda,\mu)(|\alpha\rangle \otimes | \beta\rangle ) = \sum_{\gamma,\delta \in \Z_{\ge 0}^n}\mathscr{S}(\lambda,\mu)_{\alpha,\beta}^{\gamma,\delta} |\gamma\rangle \otimes | \delta\rangle,\label{smdef}\\
\mathscr{S}(\lambda,\mu)^{\gamma,\delta}_{\alpha, \beta} = \delta^{\gamma+\delta}_{\alpha+\beta}\Phi_q(\gamma | \beta; \lambda,\mu), \label{lin}
\end{gather}
where $\Phi_q(\gamma | \beta; \lambda,\mu)$ is given by (\ref{Pdef}). We refer to the property $\mathscr{S}(\lambda,\mu)^{\gamma,\delta}_{\alpha, \beta}=0$ unless $\alpha+\beta=\gamma+\delta$ as {\em weight conservation}. The sum (\ref{smdef}) is f\/inite due to the weight conservation. In fact, the direct sum decomposition $W \otimes W = \bigoplus_{\kappa \in \Z_{\ge 0}^n} \big(\bigoplus_{\alpha+\beta=\kappa}\C |\alpha\rangle \otimes | \beta\rangle\big)$ holds
and $\mathscr{S}(\lambda,\mu)$ splits into the corresponding submatrices. Note that the ``dif\/ference property'' commonly known for the original quantum~$R$ matrix and also $S(z)$ in~(\ref{SR}) has been lost and $\mathscr{S}(\lambda,\mu) = \mathscr{S}(c\lambda,c\mu)$ does {\em not} hold.
\begin{Theorem}[\cite{KMMO}]\label{sin}
The following equalities hold:
\begin{gather}
 \text{YBE}\colon \hspace{13.8mm} \mathscr{S}_{1,2}(\nu_1,\nu_2) \mathscr{S}_{1,3}(\nu_1, \nu_3) \mathscr{S}_{2,3}(\nu_2, \nu_3) = \mathscr{S}_{2,3}(\nu_2, \nu_3) \mathscr{S}_{1,3}(\nu_1, \nu_3) \mathscr{S}_{1,2}(\nu_1,\nu_2),\!\!\!
\label{ybe2}\\
 \text{sum-to-unity}\colon \sum_{\gamma, \delta \in \Z_{\ge 0}^n} \mathscr{S}(\lambda,\mu)^{\gamma,\delta}_{\alpha, \beta} =1,\qquad \forall\, \alpha, \beta \in \Z_{\ge 0}^n. \label{sum2}
\end{gather}
\end{Theorem}

In particular (\ref{sum2}) implies
\begin{gather*}
\sum_{\gamma \in \Z_{\ge 0}^n} \Phi_q(\gamma | \beta; \lambda,\mu) = 1, \qquad\forall\, \beta \in \Z_{\ge 0}^n,
\end{gather*}
where the summands are nonzero only for $\gamma \le \beta$. Both matrices $S(z)$ and $\mathscr{S}(\lambda, \mu)$ also satisfy the inversion relation. We will depict either matrix elements $S(z)^{\gamma,\delta}_{\alpha,\beta}$ or $\mathscr{S}(\lambda,\mu)^{\gamma,\delta}_{\alpha, \beta}$ by the vertex diagram
\begin{gather}\label{vertex}
\begin{split}
& \includegraphics{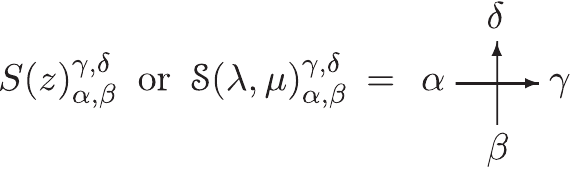}
\end{split}
\end{gather}
We note basic properties:
\begin{gather}\label{mho}
S^{m,m}(1)(|\alpha\rangle \otimes |\beta\rangle) = |\beta\rangle \otimes|\alpha\rangle,\qquad \mathscr{S}(\lambda,\lambda)(|\alpha\rangle \otimes |\beta\rangle) = |\beta\rangle \otimes|\alpha\rangle.
\end{gather}
The former originates in (\ref{wsy}) with $(l,x)=(m,y)$ and the latter can be checked easily from~(\ref{lin}) and (\ref{Pdef}).

\subsection{Commuting transfer matrices with periodic boundary condition}
\label{sec2.2}
For $l,m_1,\ldots, m_L \in \Z_{\ge 1}$ and parameters $z, w_1, \ldots, w_L$, set
\begin{gather}\label{tdef}
T\big(l,z|{\textstyle {m_1,\ldots, m_L \atop w_1,\ldots, w_L}}\big)= \operatorname{Tr}_{V_l}\big(S^{l,m_L}_{0,L}(z/w_L)\cdots S^{l,m_1}_{0,1}(z/w_1)\big)
\in \operatorname{End}\big(V_{m_1} \otimes \cdots \otimes V_{m_L}\big).
\end{gather}
In the terminology of the quantum inverse scattering method, it is the row transfer matrix of the $U_q\big(A^{(1)}_n\big)$ vertex model of size $L$ with periodic boundary condition whose quantum space is $V_{m_1} \otimes \cdots \otimes V_{m_L}$ with inhomogeneity $w_1, \ldots, w_L$ and the auxiliary space $V_l$ signif\/ied by~0 with spectral parameter~$z$. The $S^{l,m_i}_{0,i}(z/w_i)$ is the matrix~(\ref{SR}) acting as $S^{l,m_i}(z/w_i)$ on $V_l \otimes V_{m_i}$ and as identity elsewhere. The dependence on $q$ has been suppressed in the notation. It has the dif\/ference property $T\big(l,z|{\textstyle {m_1,\ldots, m_L \atop w_1,\ldots, w_L}}\big)= T\big(l, az|{\textstyle {m_1,\ldots, m_L \atop aw_1,\ldots, aw_L}}\big)$.

Thanks to the YBE (\ref{ybe1}), it forms a commuting family \cite{Bax}:
\begin{gather*}
\big[T\big(l,z|{\textstyle {m_1,\ldots, m_L \atop w_1,\ldots, w_L}}\big), T\big(l',z'|{\textstyle {m_1,\ldots, m_L \atop w_1,\ldots, w_L}}\big)\big]=0.
\end{gather*}
The $T=T(l,z|{\textstyle {m_1,\ldots, m_L \atop w_1,\ldots, w_L})}$ acts on a base vector representing a row conf\/iguration as\footnote{We warn that $|\alpha_1,\ldots, \alpha_L\rangle$ with $\alpha_i =(\alpha_{i,1},\ldots, \alpha_{i,n+1}) \in B_{m_i}$ here is dif\/ferent from the one in (\ref{BV}).}
\begin{gather}\label{tel}
T|\beta_1,\ldots, \beta_L\rangle = \sum_{\alpha_i \in B_{m_i}} T_{\beta_1,\ldots, \beta_L}^{\alpha_1,\ldots, \alpha_L} |\alpha_1,\ldots, \alpha_L\rangle
\in V_{m_1} \otimes \cdots \otimes V_{m_L}.
\end{gather}
The matrix element is given by $T_{\beta_1,\ldots, \beta_L}^{\alpha_1,\ldots, \alpha_L} = \sum\limits_{\gamma_0 \in B_l} M\big(\gamma_0| {\textstyle{\alpha_1,\ldots, \alpha_L \atop \beta_1,\ldots, \beta_L}} |\gamma_0\big)$, where the summands are def\/ined by concatenation of the diagram~(\ref{vertex}) as follows:
\begin{gather}\label{tdiag}
\begin{split}
& \includegraphics{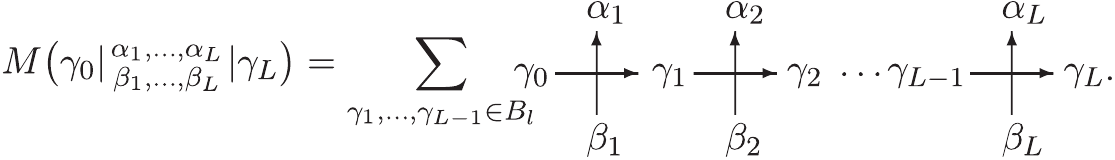}
\end{split}
\end{gather}
We have suppressed the dependence on $(l,z)$ and $(m_i,w_i)$ attached to the horizontal and vertical arrows. By the construction $T$ satisf\/ies the weight conservation:
\begin{gather}\label{wc}
T_{\beta_1,\ldots, \beta_L}^{\alpha_1,\ldots, \alpha_L} = 0
\qquad \text{unless}\quad \alpha_1+\cdots +\alpha_L = \beta_1+\cdots + \beta_L \in \Z_{\ge 0}^{n+1}.
\end{gather}

Next we proceed to the transfer matrix associated with $\mathscr{S}(\lambda,\mu)$ in (\ref{smdef}):
\begin{gather}\label{ioa}
\mathscr{T}(\lambda|\mu_1,\ldots, \mu_L) = \operatorname{Tr}_{W}\left( \mathscr{S}_{0,L}(\lambda,\mu_L)\cdots \mathscr{S}_{0,1}(\lambda,\mu_1)
\right) \in \operatorname{End}\big(W^{\otimes L}\big),
\end{gather}
where the notations are similar to (\ref{tdef}). Let $\mathscr{T}_{\beta_1,\ldots, \beta_L}^{\alpha_1,\ldots, \alpha_L}$ be its matrix element def\/ined analogously to (\ref{tel}). It is specif\/ied as $\mathscr{T}_{\beta_1,\ldots, \beta_L}^{\alpha_1,\ldots, \alpha_L}= \sum\limits_{\gamma_0 \in \Z^n_{\ge 0}}
\mathscr{M}\big(\gamma_0| {\textstyle{\alpha_1,\ldots, \alpha_L \atop \beta_1,\ldots, \beta_L}} |\gamma_0\big)$, where $\mathscr{M}\big(\gamma_0| {\textstyle{\alpha_1,\ldots, \alpha_L \atop \beta_1,\ldots, \beta_L}} |\gamma_L\big)$ is def\/ined by~(\ref{tdiag}) if the $i$-th vertex from the left is regarded as $\mathscr{S}(\lambda,\mu_i)^{\gamma_i, \alpha_i}_{\gamma_{i-1},\beta_i}$ in~(\ref{lin}) and~$\alpha_i$'s and the sum over $\gamma_1,\ldots, \gamma_{L-1}$ are taken from~$\Z_{\ge 0}^n$. The horizontal arrow should then be understood as being associated with $\lambda$ as well although it is suppressed in the notation. Since the summand vanishes unless $\gamma_i\le \beta_i$ for all~$i$, the sum (\ref{tdiag}) for $\gamma_i \in\Z_{\ge 0}^n$ is f\/inite and $\mathscr{T}(\lambda|\mu_1,\ldots, \mu_L)$ is well-def\/ined. We have the commutativity
\begin{gather}\label{kyk}
\big[ \mathscr{T}(\lambda|\mu_1,\ldots, \mu_L), \mathscr{T}(\lambda'|\mu_1,\ldots, \mu_L) \big]=0
\end{gather}
and the weight conservation analogous to (\ref{wc}). The Bethe ansatz eigenvalues of $T\big(l,z|{\textstyle {m_1,\ldots, m_L \atop w_1,\ldots, w_L}}\big)$ and
$\mathscr{T}(\lambda|\mu_1,\ldots, \mu_L)$ have been obtained in \cite[Section~4]{KMMO}.

\subsection{Commuting transfer matrices with mixed boundary condition}\label{sec2.3}
Let us present a simple example of commuting transfer matrices having mixed boundary conditions. Let $\langle \gamma|$ with $\gamma \in B_l$ be the basis of the dual of $V_l$ such that $\langle \gamma | \gamma'\rangle = \delta_{\gamma,\gamma'}$. Def\/ine
\begin{align*}
\tilde{T}\big(i,l,z|{\textstyle {m_1,\ldots, m_L \atop w_1,\ldots, w_L}}\big)
= \sum_{\gamma \in B_l}\langle \gamma |  S^{l,m_L}_{0,L}(z/w_L)\cdots S^{l,m_1}_{0,1}(z/w_1)|l {\bf e}_i\rangle
\in \operatorname{End}\left(V_{m_1} \otimes \cdots \otimes V_{m_L}\right),
\end{align*}
where $\langle \gamma |$ and $|l {\bf e}_i\rangle$ are regarded as sitting at the 0-th tensor component. The vector
$|l{\bf e}_i\rangle$ is def\/ined by (\ref{hrm}). The matrix element is given by $\tilde{T}_{\beta_1,\ldots, \beta_L}^{\alpha_1,\ldots, \alpha_L} = \sum\limits_{\gamma \in B_l} M\big(l{\bf e}_i| {\textstyle{\alpha_1,\ldots, \alpha_L \atop \beta_1,\ldots, \beta_L}} |\gamma\big)$.
In the diagram~(\ref{tdiag}) it corresponds to the f\/ixed boundary condition $\gamma_0=l{\bf e}_i$ on the left and the free boundary condition $\gamma_L=\gamma$ on the right. Schematically we have
\begin{gather*}
\includegraphics{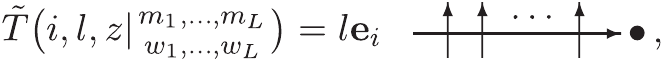}
\end{gather*}
where $\bullet$ is attached to emphasize the sum\footnote{Of course there are also many other edges where the sum is to be taken.}. It forms a commuting family:
\begin{gather*}
\big[\tilde{T}\big(i,l,z|{\textstyle {m_1,\ldots, m_L \atop w_1,\ldots, w_L}}\big), \tilde{T}\big(i,l',z'|{\textstyle {m_1,\ldots, m_L \atop w_1,\ldots, w_L}}\big)\big]=0.
\end{gather*}
This is most easily seen graphically as follows:
\begin{gather*}
\includegraphics{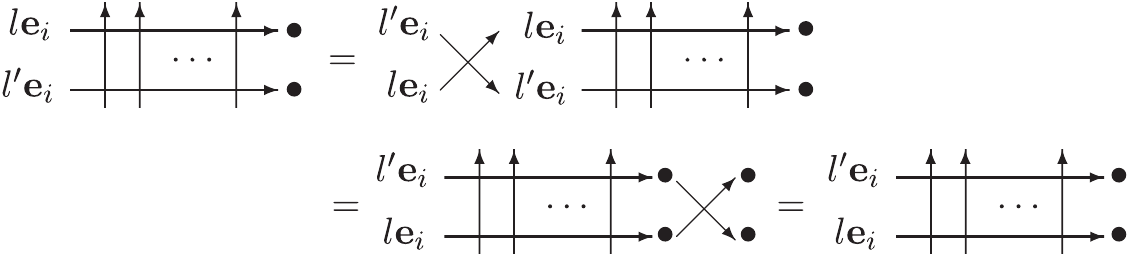}
\end{gather*}
The top left diagram depicts the product $\tilde{T}\big(i,l,z|{\textstyle {m_1,\ldots, m_L \atop w_1,\ldots, w_L}}\big)
\tilde{T}\big(i,l',z'|{\textstyle {m_1,\ldots, m_L \atop w_1,\ldots, w_L}}\big)$ and so does the bottom right (with opposite order). The f\/irst equality is by the normalization $R^{l',l}(z'/z)(|l'{\bf e}_{i}\rangle \otimes |l{\bf e}_{i}\rangle) = |l'{\bf e}_{i}\rangle \otimes |l{\bf e}_{i}\rangle$ (\ref{yuk}),
the second is by the YBE~(\ref{ybe1}) and the last is due to the sum-to-unity property (\ref{sum1}).

The analogous transfer matrix can also be constructed from $\mathscr{S}(\lambda, \mu)$ by modifying (\ref{ioa}) as
\begin{gather*}
\tilde{\mathscr{T}}(\lambda|\mu_1,\ldots, \mu_L) = \sum_{\gamma \in Z^n_{\ge 0}}\langle \gamma |
\mathscr{S}_{0,L}(\lambda,\mu_L)\cdots \mathscr{S}_{0,1}(\lambda,\mu_1) |0\rangle \in \operatorname{End}\big(W^{\otimes L}\big).
\end{gather*}
Here $|0\rangle = |0,\ldots, 0\rangle \in W$ and $\langle \gamma|$ is the basis of the dual of $W$ such that $\langle \gamma | \gamma'\rangle = \delta_{\gamma,\gamma'}$. We have the commuting family
\begin{gather*}
\big[ \tilde{\mathscr{T}}(\lambda|\mu_1,\ldots, \mu_L), \tilde{\mathscr{T}}(\lambda'|\mu_1,\ldots, \mu_L) \big]=0
\end{gather*}
by the same token owing to the normalization $\mathscr{S}(\lambda, \mu) (|0\rangle \otimes |0\rangle) =|0\rangle \otimes |0\rangle$, the YBE (\ref{ybe2}) and the sum-to-unity (\ref{sum2}).

\section{Integrable multispecies zero range process}\label{sec3}
\subsection{Discrete time Markov process}\label{sec3.1}
In the previous section we have introduced four kinds of commuting transfer matrices. Here we design their specializations to be called {\em Markov transfer matrices} that give rise to discrete time Markov processes. Denoting the time variable by~$t$ we consider the four systems endowed with the following time evolutions:
\begin{gather}
|P(t+1)\rangle = T\big(l,q^l|{\textstyle {m_1,\ldots, m_L \atop q^{m_1},\ldots, q^{m_L}}}\big) |P(t)\rangle \in V_{m_1} \otimes \cdots \otimes V_{m_L},
\label{ms1}\\
|P(t+1)\rangle = \tilde{T}\big(i,l,q^l|{\textstyle {m_1,\ldots, m_L \atop q^{m_1},\ldots, q^{m_L}}}\big)|P(t)\rangle \in V_{m_1} \otimes \cdots \otimes V_{m_L}, \label{ms2}\\
|P(t+1)\rangle = \mathscr{T}(\lambda|\mu_1,\ldots, \mu_L) |P(t)\rangle \in W^{\otimes L},\label{ms3}\\
|P(t+1)\rangle = \tilde{\mathscr{T}}(\lambda|\mu_1,\ldots, \mu_L)|P(t)\rangle \in W^{\otimes L}.\label{ms4}
\end{gather}
Let us write the action of these transfer matrices on the respective base vectors uniformly as
\begin{gather*}
T|\beta_1,\ldots, \beta_L\rangle = \sum_{\alpha_1,\ldots, \alpha_L} T_{\beta_1,\ldots, \beta_L}^{\alpha_1,\ldots, \alpha_L}|\alpha_1,\ldots, \alpha_L\rangle.
\end{gather*}
Then (\ref{ms1})--(\ref{ms4}) can be regarded as the master equation of a~Markov process with discrete time~$t$ if the following conditions are satisf\/ied:
\begin{enumerate}\itemsep=0pt
\item[(i)] non-negativity: $T_{\beta_1,\ldots, \beta_L}^{\alpha_1,\ldots, \alpha_L}\in \R_{\ge 0}$ for any $(\alpha_1, \ldots, \alpha_L)$ and $(\beta_1,\ldots, \beta_L)$,

\item[(ii)] sum-to-unity: $\sum_{\alpha_1,\ldots, \alpha_L} T_{\beta_1,\ldots, \beta_L}^{\alpha_1,\ldots, \alpha_L} = 1$ for any $(\beta_1,\ldots, \beta_L)$.
\end{enumerate}
The latter represents the total probability conservation.

\begin{Proposition}\label{yf}
The conditions $(i)$ and $(ii)$ are satisfied if $l \le \min\{m_1,\ldots, m_L\}$ and $q \in \R_{>0}$ for the systems \eqref{ms1} and \eqref{ms2}, and if
$0 < \mu^{\epsilon}_i < \lambda ^{\epsilon} < 1$, $0< q^{\epsilon}<1$ for $\epsilon=\pm 1$ for the sys\-tems~\eqref{ms3} and~\eqref{ms4}.
\end{Proposition}
The proof was given for (\ref{ms1}) and (\ref{ms3}) in~\cite{KMMO}. For (\ref{ms2}) and (\ref{ms4}), the proof is quite similar. In fact the non-negativity follows from the factorized explicit formulas (\ref{fac}), (\ref{Pdef}), (\ref{lin}), and the sum-to-unity property~(ii) does from~(\ref{sum1}) and~(\ref{sum2}).

Thus we have constructed, in the regimes specif\/ied in Proposition \ref{yf}, commuting families of discrete time Markov processes labeled with~$l$ in~(\ref{ms1}),~(\ref{ms2}) (for each $i$) and with~$\lambda$ in~(\ref{ms3}),~(\ref{ms4}). For an interpretation in terms of stochastic dynamics of multispecies particles, see \cite[Section~3.2]{KMMO}. In the systems~(\ref{ms2}) and~(\ref{ms4}), the weight, i.e., the number of particles, is not conserved.

\begin{Remark}\label{tsk}
The sum-to-unity (\ref{sum1}), (\ref{sum2}) of the stochastic $R$ matrices alone does {\em not} necessarily imply the corresponding property~(ii) for the transfer matrices if one stays in the periodic boundary condition. The latter can be established by also using the independence of the matrix elements of the NW indices in~(\ref{vertex}), i.e., $\alpha$, $\delta$ in~(\ref{fac}) and~(\ref{lin}) except the weight conservation factor $\delta_{\alpha+\beta}^{\gamma+\delta}$. See the explanation after equation (39) in~\cite{KMMO}. Thus the specialization $z=q^{l-m}$ in (\ref{fac}) achieves
the {\em double} benef\/it; the factorization manifesting the non-negativity and the `NW-freeness' making the sum-to-unity of~$R$ matrices propagate to the transfer matrices.
\end{Remark}

\begin{Remark}\label{rna}
In the evolution system (\ref{ms1}), suppose that $m_1, \ldots, m_L$ are all distinct. Set
\begin{gather*}
T_i = T\big(m_i,q^{m_i}| {\textstyle {m_1,\ldots, m_L \atop q^{m_1},\ldots, q^{m_L}}}\big) =\operatorname{Tr}_{V_{m_{i}}}\big( S^{m_i,m_L}_{0,L}\big(q^{m_{i}-m_{L}}\big) \cdots S^{m_i,m_i}_{0,i}(1)\cdots S^{m_i,m_1}_{0,1}\big(q^{m_{i}-m_{1}}\big)\big)
\end{gather*}
for $1 \le i \le L$. By applying the left relation in (\ref{mho}) to $S^{m_i,m_i}_{0,i}(1)$, one sees that the auxiliary space `merges' into the quantum space, therefore the commuting time evolutions $T_1,\ldots, T_L$ can apparently be described without the former space as illustrated by the following diagrams\footnote{In case $m_1, \ldots, m_L$ are not distinct, the commutativity still holds, but in the corresponding diagram ``par\-tic\-les'' do not come back to the original position.} for $L=3$:
\begin{gather*}
\includegraphics{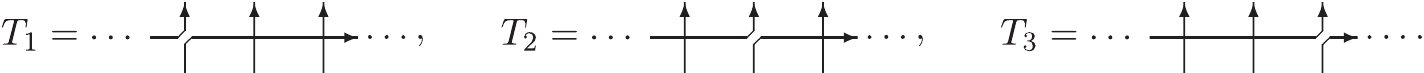}
\end{gather*}
If each arrow is formally viewed as a particle, they move around periodically to come back to the original position thereby `stirring' themselves. Such a~particle system appeared f\/irst in Yang's analysis on the 1D delta-function interaction gas \cite[equation~(14)]{Yang}. In our setting, $T_p$~with $p$ such that $m_p = \min\{m_1,\ldots, m_L\}$ fulf\/ills the conditions~(i),~(ii) in the above, which may therefore be recognized as {\em stochastic Yang's system}. The system (\ref{ms3}) also contains the similar Yang's system at $\lambda = \mu_i$ due to the right relation in~(\ref{mho}).
\end{Remark}

\subsection{Continuous time Markov process}\label{sec3.2}
In this subsection we focus on the systems (\ref{ms3}) and (\ref{ms4}) with the homogeneous choice of the parameters $\mu_1=\cdots = \mu_L=\mu$:
\begin{gather*}
\mathscr{T}(\lambda|\mu) = \mathscr{T}(\lambda|\mu,\ldots, \mu), \qquad
\tilde{\mathscr{T}}(\lambda|\mu) = \tilde{\mathscr{T}}(\lambda|\mu,\ldots, \mu).
\end{gather*}
From the special values $\Phi_q(\gamma|\beta; 1, \mu)= \delta_{\gamma,0}$ and $\Phi_q(\gamma|\beta; \mu, \mu)= \delta_{\gamma,\beta}$, one can deduce that
$\mathscr{T}(1|\mu) = \tilde{\mathscr{T}}(1|\mu) = \mathrm{id}_{W^{\otimes L}}$ and $\mathscr{T}(\mu|\mu)$ is a cyclic shift; $\mathscr{T}(\mu|\mu)|\alpha_1,\ldots, \alpha_L\rangle = |\alpha_L,\alpha_1,\ldots, \alpha_{L-1}\rangle$. Using these facts we take the logarithmic derivatives ($\epsilon=\pm 1$)
\begin{gather}
H_r = \left.-\epsilon \mu^{-1} \frac{\partial \log \mathscr{T}(\lambda | \mu)}{\partial \lambda}\right|_{\lambda=1}, \qquad
H_l = \left. \epsilon \mu \frac{\partial \log \mathscr{T}(\lambda | \mu)}{\partial \lambda} \right|_{\lambda = \mu},
\label{Hs}\\
\tilde{H} = \left.-\epsilon \mu^{-1} \frac{\partial \log \tilde{\mathscr{T}}(\lambda | \mu)}{\partial \lambda}\right|_{\lambda=1} \label{Ht}
\end{gather}
according to Baxter's formula (cf.~\cite[Section~3.4]{KMMO} and \cite[equation~(10.14.20)]{Bax}), a standard prescription to generate spin chain Hamiltonians from vertex models in equilibrium. An intriguing aspect in our setting is the presence of {\em two} such `Hamiltonian points' $\lambda=1$ and $\lambda= \mu$
for the periodic case leading to the operators expressed as the sum of local terms\footnote{As it turns out, the subscripts in $H_r$ and $H_l$ signify the right and the left hopping dynamics of particles.}:
\begin{gather}\label{sxn}
H_r= \sum_{i\in \Z_L} h_{r,i,i+1},\qquad H_l= \sum_{i\in \Z_L} h_{l,i,i+1},\qquad \tilde{H} = \sum_{1 \le i\le L-1} h_{r,i,i+1} +\tilde{h}_L.
\end{gather}
Here $h_{r,i,i+1}$, $h_{l,i,i+1}$ act as $h_r, h_l \in \operatorname{End}(W\otimes W)$ on the $i$-th and the $(i+1)$-th components and as the identity elsewhere. The $\tilde{h}_L$ acts as $\tilde{h} \in \operatorname{End} W$ on the $L$-th component and as the identity elsewhere. They are described explicitly as ($0^n=(0,\ldots, 0) \in \Z^n_{\ge 0}$)
\begin{gather}
h_r|\alpha, \beta\rangle =\epsilon\!\!\sum_{\gamma \in \Z^n_{\ge 0}\setminus \{0^n\}}\!\!
\frac{q^{\varphi(\alpha-\gamma,\gamma)} \mu^{|\gamma|-1}(q)_{|\gamma|-1}} {(\mu q^{|\alpha|-|\gamma|})_{|\gamma|}} \prod_{i=1}^n \binom{\alpha_i}{\gamma_i}_{q}|\alpha-\gamma, \beta+\gamma\rangle
-\epsilon\sum_{i=0}^{|\alpha|-1}\frac{q^i|\alpha, \beta\rangle}{1-\mu q^i},\!\!\! \label{ri1}\\
h_l|\alpha, \beta\rangle
= \epsilon \sum_{\gamma \in \Z^n_{\ge 0}\setminus \{0^n\}}
\frac{q^{\varphi(\gamma, \beta-\gamma)}
(q)_{|\gamma|-1}}
{(\mu q^{|\beta|-|\gamma|})_{|\gamma|}}
\prod_{i=1}^n
\binom{\beta_i}{\gamma_i}_{q}
|\alpha+\gamma, \beta-\gamma\rangle
-\epsilon\sum_{i=0}^{|\beta|-1}\frac{|\alpha, \beta\rangle}{1-\mu q^i},\label{ri2}\\
\tilde{h}|\alpha\rangle=
\epsilon\sum_{\gamma \in \Z^n_{\ge 0}\setminus \{0^n\}}
\frac{q^{\varphi(\alpha-\gamma,\gamma)}
\mu^{|\gamma|-1}(q)_{|\gamma|-1}}
{(\mu q^{|\alpha|-|\gamma|})_{|\gamma|}}
\prod_{i=1}^n
\binom{\alpha_i}{\gamma_i}_{q}|\alpha-\gamma\rangle
-\epsilon\sum_{i=0}^{|\alpha|-1}\frac{q^i|\alpha\rangle}{1-\mu q^i}.\label{ri3}
\end{gather}
The matrix $\tilde{h}$ has the form obtainable from (\ref{ri1}) by removing the dependence on~$\beta$.

Consider the time evolution equation
\begin{gather}\label{skb}
\frac{d}{dt}|P(t)\rangle = H |P(t)\rangle \in W^{\otimes L}, \qquad H = H_r, H_l, \tilde{H}.
\end{gather}
Denote the action of the matrices $H= H_r, H_l, \tilde{H}$ on the base vectors uniformly as $H|\beta_1,\ldots, \beta_L \rangle $ $= \sum\limits_{\alpha_1,\ldots, \alpha_L} H^{\alpha_1, \ldots, \alpha_L}_{\beta_1,\ldots, \beta_L} |\alpha_1,\ldots, \alpha_L\rangle$. The equation (\ref{skb}) can be viewed as the master equation of a~continuous time Markov process if the following conditions are satisf\/ied:
\begin{enumerate}\itemsep=0pt
\item[$(i)'$] non-negativity: $H^{\alpha_1, \ldots, \alpha_L}_{\beta_1,\ldots, \beta_L} \in \R_{\ge 0}$ for any pair such that $(\alpha_1, \ldots, \alpha_L) \neq (\beta_1,\ldots, \beta_L)$,

\item[$(ii)'$] sum-to-zero: $\sum\limits_{\alpha_1, \ldots, \alpha_L} H^{\alpha_1, \ldots, \alpha_L}_{\beta_1,\ldots, \beta_L} = 0$ for any $(\beta_1,\ldots, \beta_L)$.
\end{enumerate}
The latter represents the total probability conservation. From (\ref{ri1})--(\ref{ri3}) we note that for a f\/ixed initial condition there is a constant $M$ satisfying $|H^{\alpha_1,\ldots, \alpha_L}_{\alpha_1,\ldots, \alpha_L}|<M$ uniformly for all $(\alpha_1,\ldots, \alpha_L)$.

\begin{Proposition}[\cite{KMMO}]\label{sana}
For all the matrices $H_r$, $H_l$ and $\tilde{H}$, the condition $(ii)'$ is satisfied. So is the condition $(i)'$ if $0 < q^\epsilon, \mu^\epsilon <1$ for $\epsilon = \pm 1$.
\end{Proposition}

Thus we have obtained, in the regimes specif\/ied in Proposition \ref{sana}, the Markov processes corresponding to the continuous time variant of~(\ref{ms3}) and~(\ref{ms4}). The commutativity~(\ref{kyk}) and the construction~(\ref{Hs}) imply $[H_r, H_l]=0$. Therefore they share the same eigenvectors with the superposition $H(a,b,\epsilon,q,\mu) =a H_r(\epsilon,q,\mu)+ b H_l(\epsilon,q,\mu)$ with the coef\/f\/icients $a$, $b$ to be taken positive in the Markov process context. A curious symmetry $H(a,b,-\epsilon, q^{-1}, \mu^{-1}) = \mathscr{P}H(\mu b, \mu a,\epsilon,q,\mu)\mathscr{P}^{-1}$ is known to hold \cite[Remark~9]{KMMO}, where $\mathscr{P} = \mathscr{P}^{-1} \in \operatorname{End}(W^{\otimes L})$ is the `parity' operator reversing the sites as
$\mathscr{P}|\alpha_1,\ldots, \alpha_L\rangle = |\alpha_L, \ldots, \alpha_1\rangle$.

The Markov processes (\ref{skb}) with $H = H_r, H_l$ are naturally viewed as the stochastic dynamics of $n$-species of particles on a ring of length~$L$. The base vector $|\alpha_1, \ldots, \alpha_L\rangle$ with $\alpha_i = (\alpha_{i,1}\ldots, \alpha_{i,n}) \in \Z^n_{\ge 0}$ represents a conf\/iguration in which
there are $\alpha_{i,a}$ particles of species~$a$ at the $i$-th site. There is no constraint on the number of particles that occupy a site. The matrices~$H_r$ and $H_l$ describe the stochastic hopping of them to the right and the left nearest neighbor sites, respectively. The transition rate can be read of\/f the f\/irst terms on the r.h.s.\ of~(\ref{ri1}) and~(\ref{ri2}), where the array $\gamma=(\gamma_1,\ldots, \gamma_n)$ specif\/ies the numbers of particles that are jumping out. In case of the superposition $H(a,b,\epsilon,q,\mu)$ mentioned in the above, we have a~{\em mixture} of such right and left movers. Note that the rate is determined from the occupancy of the departure site only and independent of that at the destination site, justifying the name ``zero range process''.
Here is a snapshot of the system for the $n=2$ case\footnote{The arrangement of particles within each site does not matter.}:
\begin{gather*}
\includegraphics{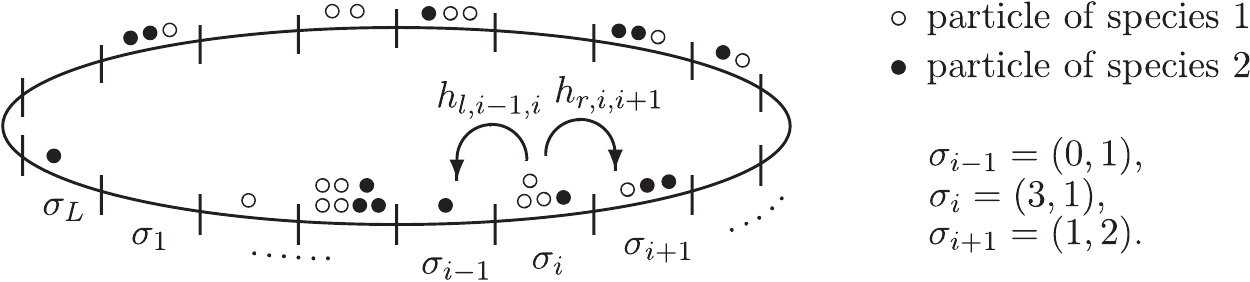}
\end{gather*}

The process (\ref{skb}) with $H=\tilde{H}$ is similarly interpreted as a stochastic particle system def\/ined on a length $L$ segment with boundaries. It consists of right movers only which will eventually exit from the right boundary:
\begin{gather*}
\includegraphics{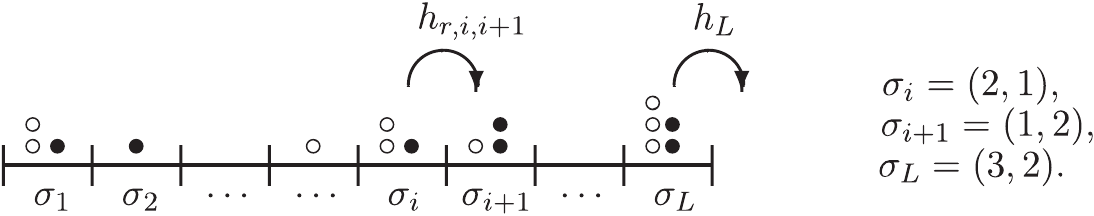}
\end{gather*}

The integrable Markov processes constructed here and Section~\ref{sec3.1} cover several models studied earlier. When $\epsilon=1, \mu\rightarrow 0$ in $H_r$,
the nontrivial local transitions in~(\ref{ri1}) are limited to the case $|\gamma| =1$. So if $\gamma_a=1$ and the other components of $\gamma$ are~0, the rate reduces to $q^{\alpha_1+\cdots + \alpha_{a-1}}\frac{1-q^{\alpha_a}}{1-q}$. This reproduces the $n$-species $q$-boson process in~\cite{T1} whose $n=1$ case further goes back to~\cite{SW2}. For $n=1$, there are extensive list of works including~\cite{BeS,BCG,BP,C,CP,GS,P,T0} for example. One can overview their interrelation in~\cite[Figs.~1 and~2]{Kuan}. When $\epsilon=1, (\mu,q)\rightarrow (0,0)$ in $H_l$, a~kinematic constraint $\varphi(\gamma, \beta-\gamma) =
\sum\limits_{1 \le i<j \le n}\gamma_i(\beta_j-\gamma_j)=0$ occurs in~(\ref{ri2}). In fact, in order that $\gamma_a>0$ happens, the equalities $\gamma_{a+1}=\beta_{a+1}, \gamma_{a+2}=\beta_{a+2}, \ldots, \gamma_n = \beta_n$ must hold. It means that larger species particles have the {\em priority} to jump out, which precisely reproduces the $n$-species totally asymmetric zero range process explored in \cite{KMO1, KMO2} after reversing the labeling of the
species $1,2, \ldots, n$ of the particles.

\subsection[Models associated with $S(z)$]{Models associated with $\boldsymbol{S(z)}$}\label{sec3.3}
Let us remark on the models associated with the stochastic $R$ matrix $S(z)=S^{l,m}(z)$ (\ref{SR}). We refer to \cite{Kuan} for a further account.

To the relevant transfer matrix (\ref{tdef}) with the uniform choice $m_i=m$ and $w_i=1$, one can associate the Hamiltonian $\mathcal{H}(m) = \pm \frac{\partial}{\partial z}\log T\big(m,z|{\textstyle {m,\ldots, m \atop 1, \ldots, 1}}\big)|_{z=1}$ similarly to (\ref{Hs}) and (\ref{sxn}). In fact one f\/inds $\mathcal{H}(m) = \sum\limits_{i\in \Z_L}h(m)_{i,i+1}$ using the left property in (\ref{mho}) where the local Hamiltonian $h(m) |\alpha, \beta\rangle = \sum\limits_{\gamma,\delta} h(m)^{\gamma,\delta}_{\alpha,\beta}|\delta,\gamma\rangle$ is specif\/ied by $h(m)^{\gamma,\delta}_{\alpha,\beta} = \pm \frac{\partial }{\partial z}S(z)^{\gamma,\delta}_{\alpha,\beta}|_{z=1}$. This is the standard derivation of an integrable `spin $\frac{m}{2}$' Hamiltonian for $U_q\big(A^{(1)}_n\big)$ by the Baxter formula (cf.~\cite[equation~(10.14.20)]{Bax}). Thanks to (\ref{sum1}), $\mathcal{H}(m)$ enjoys the sum-to-zero property~$(ii)'$ mentioned after~(\ref{skb}). However the non-negativity~$(i)'$ is {\em not} satisf\/ied just by adjusting the overall sign in general for $m \ge 2$\footnote{This is also noted in \cite[Remark~5.2]{CP} for $m=2$. On the other hand for $S^{l,m}(z)$ itself rather than its derivative, there is a~range in which elements are non-negative \cite[Proposition~3.6]{Kuan}.}. This is one of the dif\/f\/iculties that has been overcome by switching from $S^{m,m}(z)$ to $\mathscr{S}(\lambda, \mu)$. The exception is $m=1$, where one sees from Example~\ref{ykw} that $\mathcal{H}(1)$ (with the $+$ sign in the above) is nothing but the Markov matrix of the $n$-species asymmetric simple exclusion process (ASEP)\footnote{It is a model in which each site can take $n+1$ states.}. The transition rate $r_{i,j}$ of $|i,j\rangle \rightarrow |j,i\rangle$ satisf\/ies $r_{i,j}: r_{j,i} = 1: q^2$ for $1 \le i < j \le n+1$. Many results have been obtained for the $n$-species ASEP with general $n$ including, for example, matrix product stationary states~\cite{PEM}, spectral duality with respect to the Hasse diagram of sectors~\cite{AKSS}, solutions to stochastic initial value problem on inf\/inite lattice~\cite{TW}, connection to the tetrahedron equation at $q=0$~\cite{KMO0}, application to generalized Macdonald polynomials~\cite{GDW} and so on.

\section{Stationary states}\label{sec4}
\subsection{Stationary probability}\label{sec4.1}
Here we consider the systems (\ref{ms1}) and (\ref{ms3}). Introduce the f\/inite-dimensional subspaces of $V_{m_1} \otimes \cdots \otimes V_{m_L}$ and~$W^{\otimes L}$ with f\/ixed weight:
\begin{gather*}
V(k) = \oplus_{(\sigma_1,\ldots, \sigma_L) \in B(k)} \C |\sigma_1,\ldots, \sigma_L\rangle,\qquad
W(k) = \oplus_{(\sigma_1,\ldots, \sigma_L) \in \mathcal{B}(k)} \C |\sigma_1,\ldots, \sigma_L\rangle,
\end{gather*}
where the labeling sets of the bases are given by
\begin{gather}
B(k) = \big\{(\sigma_1,\ldots, \sigma_L) \in B_{m_1}\times \cdots \times B_{m_L}\,|\, \sigma_1 + \cdots + \sigma_L = k\big\},\label{hrk1}\\
\mathcal{B}(k) = \big\{(\sigma_1,\ldots, \sigma_L) \in (\Z^n_{\ge 0})^L\,|\, \sigma_1 + \cdots + \sigma_L = k\big\},\label{hrk2}
\end{gather}
where $k\in \Z^{n+1}_{\ge 0}$ in (\ref{hrk1}) and $k \in \Z^{n}_{\ge 0}$ in (\ref{hrk2}). Note that $B(k)=\varnothing$ if $|k| > m_1+\cdots + m_L$.

Denote the discrete time evolutions (\ref{ms1}) and (\ref{ms3}) simply by $|P(t+1)\rangle = T|P(t)\rangle$. They decompose into the equations on
the subspaces~$V(k)$ for (\ref{ms1}) and $W(k)$ for (\ref{ms3}), which we call {\em sectors}. If some components of $k\in \Z^n_{\ge 0}$ are~0, the system in such a sector becomes equivalent to the one with smaller $n$ by an appropriate relabeling of the species. In view of this, we shall concentrate with no loss of generality on the situation $k \in \Z_{\ge 1}^n$ which we call {\em basic} sectors.

By def\/inition stationary states are those invariant under $T$, namely $|\bar{P}\rangle$ satisfying $|\bar{P}\rangle = T|\bar{P}\rangle$. There is a unique such vector $|\bar{P}\rangle = \sum\limits_{ \sigma_1,\ldots, \sigma_L} \mathbb{P}(\sigma_1,\ldots, \sigma_L) |\sigma_1,\ldots, \sigma_L\rangle$ in each sector satisfying $\sum\limits_{ \sigma_1,\ldots, \sigma_L} \mathbb{P}(\sigma_1,\ldots, \sigma_L) = 1$. This is a consequence of the irreducibility of the sectors under~$T$ and the Perron--Frobenius theorem. The coef\/f\/icient $\mathbb{P}(\sigma_1,\ldots, \sigma_L)$ is called the stationary probability. In what follows we shall abuse the terminology also to mean the unnormalized probabilities and states. We will treat the discrete time processes only since they cover the continuous time case. Thanks to the commutativity of the Markov transfer matrices, the stationary states are independent of~$l$ for~(\ref{ms1}) and of~$\lambda$ for~(\ref{ms3}).

Before closing the subsection, we include comments on the systems (\ref{ms4}) and (\ref{skb}) with $H = \tilde{H}$, where the number of particles is not preserved. Starting from any initial condition, any state tends to the trivial one $|0^n\rangle \otimes \cdots \otimes | 0^n \rangle$, although the relaxation to it remains as an important problem. The clue to investigating it is the Bethe eigenvalues of $\tilde{H}$. To construct them is a feasible task as done in \cite[Section~4]{KMMO} for the periodic case $H_r$ and $H_l$. We leave it for a future study.

\subsection{Some examples}\label{sec4.2}
From now on we focus on the discrete time process def\/ined by (\ref{ms3}). The stationary state $|\bar{P}\rangle$ in the sector~$W(k)$ is characterized by
\begin{gather*}
|\bar{P}\rangle = \mathscr{T}(\lambda|\mu_1,\ldots, \mu_L) |\bar{P}\rangle \in W(k).
\end{gather*}
We will refer to a sector $W(k)$ also by $k=(k_1,\ldots, k_n) \in \Z^n_{\ge 0}$ for simplicity.

It was known in the single species case $n=1$ that the stationary state possesses the product measure (cf.~\cite{EMZ, P} at least for the homogeneous case
$\forall\, \mu_i=\mu$):
\begin{gather*}
\mathbb{P}(\sigma_1,\ldots, \sigma_L) = \prod_{i=1}^L g_{\sigma_i}(\mu_i), \qquad \sigma_i \in \Z_{\ge 0},
\end{gather*}
where $g_{\sigma_i}(\mu_i)$ is the $n=1$ case of the function
\begin{gather}
g_\alpha(\mu)=\frac{\mu^{-|\alpha|}(\mu)_{|\alpha|}} {\prod\limits_{i=1}^n(q)_{\alpha_i}}\qquad \text{for} \quad \alpha = (\alpha_1,\ldots, \alpha_n) \in \Z^n_{\ge 0}. \label{mgd}
\end{gather}
Such a factorization, however, is no longer valid in the multispecies case $n\ge 2$ making the system nontrivial and interesting even without an introduction (cf.~\cite{DEHP}) of a reservoir. Particles of a given species must behave under the inf\/luence of the other species acting as a nontrivial dynamical background.

\begin{Example}\label{ex:kan}
For $L=2$, $n=2$ and the sector $k=(1,1)$, we have
\begin{gather*}
|\overline{P}\rangle = \mu _1^2 (1-\mu _2)(1-q\mu _2) (\mu _1+\mu _2-2 \mu _2 \mu _1) |\varnothing,12\rangle\\
\hphantom{|\overline{P}\rangle =}{} +\mu _1 \mu _2 (1-\mu _1) (1-\mu _2) (\mu _1+ q\mu_2-\mu _1\mu _2 -q\mu _1 \mu _2) |1,2\rangle + \mathrm{cyclic}.
\end{gather*}
Here $ |\varnothing,12\rangle$ and $ |1,2\rangle$ are the multiset representations of $|(0,0), (1,1)\rangle$ and $|(1,0), (0,1)\rangle$. For $L=3, n=2$ and the sector $k=(1,1)$, we have
\begin{gather*}
|\overline{P}\rangle = \mu_1^2 \mu _2^2 (1-\mu _3)(1-q\mu _3) (\mu _1 \mu _2+\mu _1 \mu _3+\mu _2 \mu _3-3 \mu _1 \mu _3 \mu _2 )
|\varnothing,\varnothing,12)\\
\hphantom{|\overline{P}\rangle =}{}+ \mu _1^2 \mu _2 \mu_3 (1-\mu _2) (1-\mu _3) (q\mu _1 \mu _2+\mu _1 \mu _3+ \mu _2 \mu _3 -2\mu _1 \mu _2 \mu _3-q \mu _1 \mu _2 \mu _3 ) |\varnothing,2,1\rangle\\
\hphantom{|\overline{P}\rangle =}{} +\mu _1^2 \mu _2 \mu_3
(1-\mu _2) (1-\mu _3) (\mu _1 \mu _2+q\mu _1 \mu _3+ q\mu _2 \mu _3 -\mu _1 \mu _2 \mu _3-2q \mu _1 \mu _2 \mu _3 )|\varnothing,1,2\rangle\\
\hphantom{|\overline{P}\rangle =}{}+ \mathrm{cyclic}.
\end{gather*}
Here $\mathrm{cyclic}$ means the sum of terms obtained by the replacement $\mu_j\rightarrow \mu_{j+i}$ and $|\sigma_1,\ldots, \sigma_L\rangle \rightarrow
|\sigma_{1+i},\ldots, \sigma_{L+i}\rangle$ over $i \in \Z_L$ with $i \neq 0$.
\end{Example}

Based on computer experiments we propose
\begin{Conjecture}\label{c:aim}
For the stationary state $|\bar{P}\rangle$ in any sector~$k$, there is a normalization such that $\mathbb{P}(\sigma_1,\ldots, \sigma_L) \in \Z_{\ge 0}[q,-\mu_1,\ldots, -\mu_L]$ for all $(\sigma_1,\ldots, \sigma_L) \in \mathcal{B}(k)$.
\end{Conjecture}

\subsection{Matrix product construction}\label{sec4.3}
We seek the steady state probability in the matrix product form
\begin{gather}\label{mst}
{\mathbb P}(\sigma_1,\ldots, \sigma_L) = \operatorname{Tr}(X_{\sigma_1}(\mu_1)\cdots X_{\sigma_L}(\mu_L))
\end{gather}
in terms of some operator $X_\alpha(\mu)$ with $\alpha \in \Z^n_{\ge 0}$. Our strategy is to invoke the following result.

\begin{Proposition}\label{pr:srn}
Suppose the operators $X_\alpha(\mu)$ $(\alpha \in \Z^n_{\ge 0})$ obey the relation
\begin{gather}\label{msk}
X_\alpha(\mu)X_\beta(\lambda) = \sum_{\gamma,\delta \in \Z^n_{\ge 0}} \mathscr{S}(\lambda, \mu)^{\beta, \alpha}_{\gamma,\delta} X_\gamma(\lambda)X_\delta(\mu).
\end{gather}
Suppose further that $\dim \operatorname{Ker}(\mathscr{T}(\mu_i|\mu_1,\ldots, \mu_L)-1)=1$ for some $i$. Then the matrix product formula~\eqref{mst} of the stationary probability holds for the system~\eqref{ms3} provided that the trace is convergent and not identically zero.
\end{Proposition}

\begin{proof}
The stationarity $\mathscr{T}(\lambda|\mu_1,\ldots, \mu_L)|\bar{P}\rangle = |\bar{P}\rangle$ follows from the special case $\lambda=\mu_i$. To see this note that the latter leads to $\mathscr{T}(\lambda|\mu_1,\ldots, \mu_L)|\bar{P}\rangle = \mathscr{T}(\lambda|\mu_1,\ldots, \mu_L) \mathscr{T}(\mu_i|\mu_1,\ldots, \mu_L)|\bar{P}\rangle =\mathscr{T}(\mu_i|\mu_1,\ldots, \mu_L) \mathscr{T}(\lambda|\mu_1,\ldots, \mu_L)|\bar{P}\rangle$, telling that $\mathscr{T}(\lambda|\mu_1,\ldots, \mu_L)|\bar{P}\rangle \in \operatorname{Ker}(\mathscr{T}(\mu_i|\mu_1,\ldots, \mu_L)-1)$. Therefore from the assumption we have $\mathscr{T}(\lambda|\mu_1,\ldots, \mu_L)|\bar{P}\rangle = f(\lambda)|\bar{P}\rangle$ for some scalar~$f(\lambda)$. Taking the ``column sum" on the both sides using the sum-to-unity property of $\mathscr{T}(\lambda|\mu_1,\ldots, \mu_L)$ we f\/ind $f(\lambda)=1$.

In the sequel, we assume $i=L$ without losing generality in the light of cyclicity of trace. Now the equality $|\bar{P}\rangle = \mathscr{T}(\mu_L|\mu_1,\ldots, \mu_L)|\bar{P}\rangle$ to be shown takes the form
\begin{gather*}
\operatorname{Tr}(X_{\alpha_1}(\mu_1)\cdots X_{\alpha_L}(\mu_L)) =\sum_{\beta_1,\ldots, \beta_L} \mathscr{T} (\mu_L|\mu_1,\ldots, \mu_L)^{\alpha_1,\ldots, \alpha_L}_{\beta_1,\ldots, \beta_L} \operatorname{Tr}(X_{\beta_1}(\mu_1)\cdots X_{\beta_L}(\mu_L)).
\end{gather*}
In the l.h.s., starting from $X_{\alpha_{L-1}}(\mu_{L-1})X_{\alpha_L}(\mu_L)$ apply the relation~(\ref{msk}) repeatedly to send $X_{\alpha_L}(\mu_L)$
to the left through the whole product to bring it back to the original rightmost position by using the cyclicity of the trace. It results in the relation
\begin{gather*}
\operatorname{Tr}(X_{\alpha_1}(\mu_1)\cdots X_{\alpha_L}(\mu_L)) = \sum_{\beta_1,\ldots, \beta_L} \mathscr{M}\big(\beta_L|
{\textstyle{\alpha_1,\ldots, \alpha_{L-1} \atop \beta_1,\ldots, \beta_{L-1}}} |\alpha_L\big) \operatorname{Tr}(X_{\beta_1}(\mu_1)\cdots X_{\beta_L}(\mu_L))
\end{gather*}
in terms of $\mathscr{M}$ def\/ined under (\ref{ioa}). It is depicted as~(\ref{tdiag}). Since the parameters attached to the horizontal arrow and the rightmost vertical arrow are both~$\mu_L$, we have $\mathscr{M}\big(\beta_L| {\textstyle{\alpha_1,\ldots, \alpha_{L-1} \atop \beta_1,\ldots, \beta_{L-1}}}
|\alpha_L\big) = \mathscr{T} (\mu_L|\mu_1,\ldots, \mu_L)^{\alpha_1,\ldots, \alpha_L}_{\beta_1,\ldots, \beta_L}$ from the right relation in (\ref{mho}).
\end{proof}

The sum-to-unity property assures $\dim \operatorname{Ker}(\mathscr{T}(\mu_i|\mu_1,\ldots, \mu_L)-1)\ge 1$ in general. We expect that this is~1 when $\mu_1,\ldots, \mu_L$ are distinct\footnote{It is greater than 1 for example when $\mu_1=\cdots = \mu_L$.}. Anyway the matrix product formula~(\ref{mst}) has been proved in \cite[Proposition~6]{KO1} without relying on the assumption $\dim \operatorname{Ker}(\mathscr{T}(\mu_i|\mu_1,\ldots, \mu_L)-1) = 1$
but using the {\em auxiliary condition} \cite[equation~(30)]{KO1} of $\mathscr{S}(\lambda, \mu)$. We have included Proposition~\ref{pr:srn} as a~possible variant avoiding such a specif\/ic property of the stochastic $R$ matrix.

\section{Zamolodchikov--Faddeev algebra}\label{sec5}
\subsection{General remarks}\label{sec5.1}
Let us write (\ref{msk}) symbolically as
\begin{gather}\label{sra}
X(\mu) \otimes X(\lambda) = \check{\mathscr{S}}(\lambda, \mu)\bigl[X(\lambda) \otimes X(\mu)\bigr],
\end{gather}
where $X(\mu) = (X_\alpha(\mu))$ denotes a collection of operators. The quadratic relation of this form is called Zamolodchikov--Faddeev (ZF) algebra.
Its associativity is guaranteed by the YBE satisf\/ied by the structure function $\check{\mathscr{S}}(\lambda, \mu)$\footnote{$\check{\mathscr{S}}(\lambda, \mu) = P \mathscr{S}(\lambda, \mu)$ with $P (|\alpha\rangle \otimes |\beta\rangle) = |\beta\rangle \otimes | \alpha\rangle$ so that $\mathscr{S}(\lambda, \mu)^{\beta, \alpha}_{\gamma,\delta} = \check{\mathscr{S}}(\lambda, \mu)^{\alpha,\beta}_{\gamma,\delta}$.}. Before presenting the specif\/ic results to our ZRP in the next subsection, we review its background brief\/ly in this subsection. Similar contents can also be found for example in \cite{CDW, CRV, GDW, SW1} and references therein.

ZF algebra was originally introduced in the integrable quantum f\/ield theory in $(1+1)$ dimension to encode the factorized scattering of particles \cite{F, ZZ}. The structure function therein should be a~properly normalized scattering matrix satisfying unitarity to guarantee the total probability conservation in the quantum f\/ield theoretical setting. In the realm of integrable Markov processes, the situation is parallel. The ZF algebra serves as a local version of the stationary condition in the matrix product construction of the stationary states. The structure function, stochastic $R$ matrix, should fulf\/ill the sum-to-unity property. It was demonstrated in the proof of Proposition~\ref{pr:srn} how the ZF algebra leads to the stationary condition of the system in the case of a discrete time Markov process. Historically, however, such quadratic relations were utilized earlier in continuous time models as `cancellation mechanism' or `hat relations'~\cite{DEHP}. In the present set-up it reads
\begin{gather}\label{nmi}
h \bigl[X(\mu) \otimes X(\mu)\bigr] = X(\mu) \otimes X'(\mu) - X'(\mu)\otimes X(\mu)
\end{gather}
with $h = \frac{\partial}{\partial \lambda} \check{\mathscr{S}}(\lambda, \mu)|_{\lambda=\mu}$. This is the derivative of~(\ref{sra}) at $\lambda = \mu$ with the additional condition $\check{\mathscr{S}}(\mu, \mu) = \mathrm{id}$ which matches~(\ref{mho}). In terms of components, it reads $\sum\limits_{\gamma,\delta}h^{\alpha,\beta}_{\gamma,\delta} X_\gamma(\mu)X_\delta(\mu) = X_\alpha(\mu)X'_\beta(\mu) -X'_\alpha(\mu)X_\beta(\mu)$ if the action of $h$ is set as $h (|\alpha\rangle \otimes |\beta\rangle) = \sum\limits_{\gamma,\delta}h^{\gamma,\delta}_{\alpha, \beta} |\gamma\rangle \otimes |\delta\rangle$. The rela\-tion~(\ref{nmi}) manifestly tells that the matrix product states $\operatorname{Tr}(X(\mu) \otimes \cdots \otimes X(\mu))$ are null vectors of the operator $H = \sum_{i \in \Z_L}h_{i,i+1}$ in the periodic setting. In retrospect, one may compare the relation between~(\ref{nmi}) and~(\ref{sra}) with that between the XXZ model and the six-vertex model in the light of Baxter's formula~(\ref{Hs}).

Back to the ZF algebra itself, it is naturally embedded into the so-called $RLL = LLR$ relation
\begin{gather}\label{rsmt}
\bigl[\mathcal{L}(\lambda) \otimes \mathcal{L}(\mu)\bigr] \check{\mathscr{S}}(\lambda, \mu) =
\check{\mathscr{S}}(\lambda, \mu)\bigl[\mathcal{L}(\lambda) \otimes \mathcal{L}(\mu)\bigr]
\end{gather}
for an $L$ operator $\mathcal{L}(\lambda)$ if there is a special index, say $0$, such that $\mathscr{S}(\lambda, \mu)^{\beta, \alpha}_{0,0}= \mathscr{S}(\lambda, \mu)_{\beta, \alpha}^{0,0} = \theta(\alpha=\beta=0)$. In fact the matrix element of~(\ref{rsmt}) for $|0\rangle \otimes |0 \rangle \rightarrow |\alpha\rangle \otimes |\beta\rangle$ gives (\ref{sra}) by the identif\/ication $\mathcal{L}(\mu)_{\alpha, 0} = X_\alpha(\mu)$. In view of this, construction of stationary states is elevated and embedded into that of representations of the stochastic $RLL=LLR$ relation in which the relevant components of $\operatorname{Tr}(\mathcal{L}(\mu_1) \otimes \cdots \otimes \mathcal{L}(\mu_L))$ are convergent and not identically zero. When the structure function is the $R$ matrix~$R^{1,1}$ of the vector representation, a universal~$L$ has been provided in~\cite{J}. Starting from it, one can cope with $RLL=LLR$ with higher $R^{l,m}$ by the corresponding fusion of $L$'s in principle~\cite{KRS}. Modifying it so as to f\/it the stochastic~$S^{l,m}$ should be feasible by an appropriate twist (cf.~\cite{GDW}). In this sense there is a standard route to achieve the matrix product construction for $S^{l,m}$-based models at least conceptually if not practically. On the other hand, a further intriguing feature is expected when the structure function is $\mathscr{S}(\lambda,\mu)$ due to the peculiarity of its origin (\ref{fac})--(\ref{lin}). This is one of our motivations in Sections \ref{sec5.2} and \ref{sec5.3}.

Turning to stationary probabilities, the maneuver in the proof of Proposition \ref{pr:srn} elucidates that it is a part of more general problem of f\/inding solutions to a quantum Knizhnik--Zamolodchikov type equation~\cite{FR} with appropriate subsidiary conditions. Such wider problems have not been addressed for our stochastic $R$ matrix~(\ref{lin}). Implication of the sum-to-unity property (\ref{sum2}) to the ZF algebra will be explained after Remark~\ref{askg} until the end of Section~\ref{sec5}. In the next subsection we will be concerned with a particular representation of the ZF algebra in terms of {\em $q$-bosons}.

\subsection[$q$-boson representation]{$\boldsymbol{q}$-boson representation}\label{sec5.2}
Let us present a $q$-boson representation of the ZF algebra~(\ref{msk})\footnote{In this paper we do not attempt to classify the representations
such that r.h.s.\ of \eqref{mst} is f\/inite and nonzero.}. Here and in the next subsection we stay in the regime $0 < q < 1$. From~(\ref{lin}) it has the explicit form
\begin{gather}\label{zz}
X_\alpha(\mu)X_\beta(\lambda) = \sum_{\gamma \le \alpha}\Phi_q(\beta|\alpha+\beta-\gamma;\lambda, \mu)
X_{\gamma}(\lambda)X_{\alpha+\beta-\gamma}(\mu),
\end{gather}
where the omitted condition $\gamma \in \Z_{\ge 0}^n$ should always be taken for granted. We f\/ind it convenient to work also with another normalization $Z_\alpha(\mu)$ specif\/ied by
\begin{gather}\label{aim}
X_\alpha(\mu) = g_\alpha(\mu)Z_\alpha(\mu),
\end{gather}
where $g_\alpha(\mu)$ has been def\/ined in (\ref{mgd}). The ZF algebra for the latter takes the form
\begin{gather}
 Z_\alpha(\mu)Z_\beta(\lambda) = \sum_{\gamma \le \alpha} q^{\varphi(\alpha-\gamma, \beta-\gamma)} \Phi_q(\gamma|\alpha; \lambda,\mu)
Z_\gamma(\lambda)Z_{\alpha+\beta-\gamma}(\mu)\label{mrn}
\end{gather}
due to the identity
\begin{gather*}
\frac{g_\gamma(\lambda)g_{\alpha+\beta-\gamma}(\mu)} {g_\alpha(\mu)g_{\beta}(\lambda)} \Phi_q(\beta|\alpha+\beta-\gamma;\lambda, \mu)
=q^{\varphi(\alpha-\gamma, \beta-\gamma)} \Phi_q(\gamma|\alpha;\lambda, \mu).
\end{gather*}

Let $\mathcal{B}$ be the algebra generated by $1$, $\bb$, $\bc$, $\bk$ obeying the relations
\begin{gather}\label{akn}
\bk \bb = q\bb \bk,\qquad \bk\bc = q^{-1} \bk \bc,\qquad \bb \bc = 1 - \bk,\qquad \bc \bb = 1-q\bk.
\end{gather}
We call it the $q$-boson algebra. It has a basis $\{\bb^i \bc^j\,|\, i,j \in \Z_{\ge 0}\}$.

Let $F = \bigoplus_{m \ge 0}\C(q) |m\rangle$ be the Fock space and $F^\ast = \bigoplus_{m \ge 0}\C(q) \langle m |$ be its dual on which the $q$-boson operators $\bb$, $\bc$, $\bk$ act as
\begin{alignat}{4}
& \bb | m \rangle = |m+1\rangle,\qquad && \bc | m \rangle = (1-q^m)|m-1\rangle, \qquad && \bk |m\rangle = q^m |m \rangle,& \nonumber\\
& \langle m | \bc = \langle m+1 |,\qquad && \langle m | \bb = \langle m-1|(1-q^m),\qquad && \langle m | \bk = \langle m | q^m,& \label{yrk}
\end{alignat}
where $|{-}1\rangle = \langle -1 |=0$ and $1$ acts as the identity. They satisfy the def\/ining relations~(\ref{akn}). The bilinear pairing of $F^\ast$ and $F$ is specif\/ied as $\langle m | m'\rangle = \delta_{m,m'}(q)_m$. Then $\langle m| (X|m'\rangle) = (\langle m|X)|m'\rangle$ is valid and the trace is given by $\operatorname{Tr}(X) = \sum\limits_{m \ge 0} \frac{\langle m|X|m\rangle}{(q)_m}$. As a vector space, the $q$-boson algebra~$\mathcal{B}$ has the direct sum decomposition $\mathcal{B} = \C(q) 1 \oplus \mathcal{B}_{\text{f\/in}}$, where $\mathcal{B}_{\text{f\/in}} = \bigoplus_{r \ge 1} (\mathcal{B}_+^r \oplus \mathcal{B}_-^r \oplus \mathcal{B}_0^r)$ with $\mathcal{B}^r_+ =\bigoplus_{s\ge 0} \C(q) \bk^s\bb^r$, $\mathcal{B}^r_- =\bigoplus_{s\ge 0} \C(q) \bk^s\bc^r$ and $\mathcal{B}^r_0 =\C(q) \bk^r$. The trace $\operatorname{Tr}(X)$ is convergent if $X \in \mathcal{B}_{\text{f\/in}}$. It vanishes unless $X \in \bigoplus_{r \ge 1}\mathcal{B}^r_0$ when it is evaluated by $\operatorname{Tr}(\bk^r) = (1-q^r)^{-1}$.

By the $q$-boson representation we mean the algebra homomorphisms
\begin{gather*}
\text{ZF algebra (\ref{zz}) or (\ref{mrn})} \rightarrow \mathcal{B}^{\otimes n(n-1)/2} \rightarrow \operatorname{End}\big(F^{\otimes n(n-1)/2}\big).
\end{gather*}
The right arrow is already given by (\ref{yrk}). In the sequel we will focus on the left arrow and denote $Z_\alpha(\mu) \mapsto (\cdots)$ simply by
$Z_\alpha(\mu) = (\cdots)$.

The ZF algebra (\ref{mrn}) admits a ``trivial'' representation $Z_\alpha(\zeta) = K_\alpha$ in terms of an opera\-tor~$K_\alpha$ satisfying $K_0 = 1$ and
$K_\alpha K_\beta = q^{\varphi(\alpha, \beta)}K_{\alpha+\beta}$ \cite[Proposition~7]{KO1}, where $\varphi(\alpha, \beta)$ is def\/ined in~(\ref{Pdef}). Such a~$K_\alpha$ is easily constructed, for instance as\footnote{We write $K_\alpha$ with $\alpha=(\alpha_1,\ldots, \alpha_n)$ as $K_{\alpha_1,\ldots, \alpha_n}$ rather than $K_{(\alpha_1,\ldots, \alpha_n)}$ for simplicity. A similar convention will also be used for $g_\alpha(\zeta), X_\alpha(\zeta)$ and $Z_\alpha(\zeta)$.}
\begin{gather}\label{K}
K_{\alpha_1,\ldots, \alpha_{n}}=\bk^{\alpha^+_1} \bc^{\alpha_1} \otimes \cdots \otimes \bk^{\alpha^+_{n-1}}\bc^{\alpha_{n-1}} \in \mathcal{B}^{\otimes n-1}, \qquad \alpha^+_i := \alpha_{i+1}+\cdots + \alpha_{n}.
\end{gather}
However this representation does not contain a creation operator $\bb$ therefore leads to vanishing trace in the matrix product formula~(\ref{mst}). Our $Z_\alpha(\zeta)$ given below may be regarded as a~perturbation series starting from the trivial representation in terms of creation operators.

For $\alpha_i \in \Z_{\ge 0}$, def\/ine the element $Z_{\alpha_1,\ldots, \alpha_{n}}(\zeta) \in \mathcal{B}^{\otimes n(n-1)/2}$ from the $n=1$ case and the recursion with respect to $n$ as follows:
\begin{gather}
Z_{\alpha_1}(\zeta) =1, \label{ini}\\
Z_{\alpha_1,\ldots, \alpha_n} (\zeta) = \sum_{l=(l_1,\ldots, l_{n-1}) \in \Z_{\ge 0}^{n-1}}X_l(\zeta) \otimes \bb^{l_1}\bk^{\alpha^+_1}\bc^{\alpha_1}\otimes \cdots \otimes \bb^{l_{n-1}}\bk^{\alpha^+_{n-1}}\bc^{\alpha_{n-1}},\label{mdk}
\end{gather}
where $X_l(\zeta)= g_l(\zeta)Z_l(\zeta)$ as in (\ref{aim}) and $\alpha^+_i$ is def\/ined by $(\ref{K})$.

\begin{Theorem}[\cite{KO2}]\label{th:main}
The $Z_\alpha(\zeta)$ defined by \eqref{ini}, \eqref{mdk} satisfies the ZF algebra \eqref{mrn} for general~$n$.
\end{Theorem}

\subsection{Explicit formula}\label{sec5.3}

From (\ref{mdk}) $Z_\alpha(\zeta)$ depends on $\alpha$ simply as
\begin{gather}\label{yry}
Z_{\alpha_1,\ldots, \alpha_{n}}(\zeta) = Z_{0^{n}}(\zeta) \bigl(1^{\otimes \frac{1}{2}(n-1)(n-2)} \otimes K_{\alpha_1,\ldots, \alpha_{n}}\bigr),
\end{gather}
in terms of $K_{\alpha_1,\ldots, \alpha_{n}}$ in~(\ref{K}). Thus it suf\/f\/ices to calculate the special case of~(\ref{mdk}):
\begin{gather}\label{zgz}
Z_{0^n}(\zeta) = \sum_{l_1,\ldots, l_{n-1} \in \Z_{\ge0}} g_{l_1,\ldots, l_{n-1}}(\zeta)Z_{l_1, \ldots, l_{n-1}}(\zeta) \otimes
\bb^{l_1} \otimes \cdots \otimes \bb^{l_{n-1}}.
\end{gather}
Utilizing
$\frac{(zw)_\infty}{(z)_\infty} = \sum\limits_{j \ge 0}\frac{(w)_j}{(q)_j}z^j$, we get the explicit formula for $n=2$:
\begin{gather}
Z_{0,0}(\zeta) = \sum_{l_1\ge 0} \frac{(\zeta)_{l_1}\zeta^{-l_1}}{(q)_{l_1}}
\bb^{l_1} = \frac{(\bb)_\infty}{(\zeta^{-1}\bb)_\infty},\qquad
Z_{\alpha_1, \alpha_2}(\zeta) = Z_{0,0}(\zeta) K_{\alpha_1,\alpha_2}=
\frac{(\bb)_\infty}{(\zeta^{-1}\bb)_\infty}\bk^{\alpha_2}
\bc^{\alpha_1},\nonumber\\
X_{\alpha_1, \alpha_2}(\zeta) = g_{\alpha_1, \alpha_2}(\zeta)Z_{\alpha_1, \alpha_2}(\zeta)
=\frac{\zeta^{-\alpha_1-\alpha_2}(\zeta)_{\alpha_1+\alpha_2}} {(q)_{\alpha_1}(q)_{\alpha_2}}
\frac{(\bb)_\infty}{(\zeta^{-1}\bb)_\infty}\bk^{\alpha_2}
\bc^{\alpha_1}. \label{Xw}
\end{gather}
It is a good exercise to conf\/irm Example \ref{ex:kan} by substituting these results into the matrix product formula~(\ref{mst}).

For $n=3$, the sum (\ref{zgz}) is calculated by using $(\zeta)_{l_1+l_2} = (\zeta)_{l_2}(q^{l_2}\zeta)_{l_1}$ as
\begin{gather*}
Z_{0,0,0}(\zeta) = \sum_{l_1,l_2}\frac{\zeta^{-l_1-l_2}(\zeta)_{l_1+l_2}} {(q)_{l_1}(q)_{l_2}}
\frac{(\bb)_\infty}{(\zeta^{-1}\bb)_\infty}\bk^{l_2} \bc^{l_1}\otimes \bb^{l_1} \otimes \bb^{l_2} \nonumber\\
\hphantom{Z_{0,0,0}(\zeta)}{} =\frac{(\bb \otimes 1 \otimes 1)_\infty}
{(\zeta^{-1}\bb \otimes 1 \otimes 1)_\infty}
\sum_{l_2} \frac{\zeta^{-l_2}(\zeta)_{l_2}(\bk \otimes 1 \otimes \bb)^{l_2}}{(q)_{l_2}}
\sum_{l_1} \frac{\zeta^{-l_1}(q^{l_2}\zeta)_{l_1}(\bc \otimes \bb \otimes 1)^{l_1}}{(q)_{l_1}}\nonumber\\
\hphantom{Z_{0,0,0}(\zeta)}{} =\frac{(\bb \otimes 1 \otimes 1)_\infty}
{(\zeta^{-1}\bb \otimes 1 \otimes 1)_\infty}
\sum_{l_2} \frac{\zeta^{-l_2}(\zeta)_{l_2}(\bk \otimes 1 \otimes \bb)^{l_2}}{(q)_{l_2}}
\frac{(q^{l_2}\bc \otimes \bb \otimes 1)_\infty}
{(\zeta^{-1}\bc \otimes \bb \otimes 1)_\infty}
\nonumber\\
\hphantom{Z_{0,0,0}(\zeta)}{} =\frac{(\bb \otimes 1 \otimes 1)_\infty}
{(\zeta^{-1}\bb \otimes 1 \otimes 1)_\infty}
(\bc \otimes \bb \otimes 1)_\infty
\sum_{l_2} \frac{\zeta^{-l_2}(\zeta)_{l_2}
(\bk \otimes 1 \otimes \bb)^{l_2}}{(q)_{l_2}}
\frac{1}{(\zeta^{-1}\bc \otimes \bb \otimes 1)_\infty}\nonumber\\
\hphantom{Z_{0,0,0}(\zeta)}{} =\frac{(\bb \otimes 1 \otimes 1)_\infty}
{(\zeta^{-1}\bb \otimes 1 \otimes 1)_\infty}
(\bc \otimes \bb \otimes 1)_\infty
\frac{(\bk \otimes 1 \otimes \bb)_\infty}
{(\zeta^{-1}\bk \otimes 1 \otimes \bb)_\infty}
\frac{1}{(\zeta^{-1}\bc \otimes \bb \otimes 1)_\infty},
\nonumber\\
Z_{\alpha_1,\alpha_2, \alpha_3}(\zeta) =
Z_{0,0,0}(\zeta) (1 \otimes K_{\alpha_1,\alpha_2,\alpha_3}) = Z_{0,0,0}(\zeta)(1 \otimes \bk^{\alpha_2+\alpha_3}\bc^{\alpha_1} \otimes
\bk^{\alpha_3} \bc^{\alpha_2}).
\end{gather*}
These results on $Z_{0,0}(\zeta)$ and $Z_{0,0,0}(\zeta)$ are summarized as\footnote{The operators $V_j(\zeta)$'s appearing in the sequel have nothing to do with that in~(\ref{BV}).}
\begin{gather}
Z_{0,0}(\zeta) = V_1(1)V_1(\zeta)^{-1},\qquad V_1(\zeta) = (\zeta^{-1}\bb)_\infty,\nonumber\\
Z_{0,0,0}(\zeta) =\bigl(Z_{0,0}(\zeta)\otimes 1\otimes 1\bigr)V_2(1)V_2(\zeta)^{-1},\nonumber\\
V_2(\zeta) = (\zeta^{-1}\bc \otimes \bb \otimes 1)_\infty (\zeta^{-1}\bk \otimes 1 \otimes \bb)_\infty.\label{knh}
\end{gather}

Now we proceed to general $n \ge 2$ case. Substitution of $(\ref{yry})|_{n\rightarrow n-1}$ into the r.h.s.\ of (\ref{zgz}) gives
\begin{gather}
Z_{0^{n}}(\zeta)= \bigl(Z_{0^{n-1}}(\zeta) \otimes 1^{\otimes n-1}\bigr)Y_n(\zeta),\nonumber\\
Y_n(\zeta)=\sum_{l_1,\ldots, l_{n-1} \in \Z_{\ge 0}}g_{l_1,\ldots, l_{n-1}}(\zeta)
1^{\otimes \frac{1}{2}(n-2)(n-3)}\otimes K_{l_1,\ldots, l_{n-1}}\otimes \bb^{l_1} \otimes \cdots \otimes \bb^{l_{n-1}}.\label{ain}
\end{gather}
It is handy to describe $\mathcal{B}^{\otimes n(n-1)/2}$ by introducing the copies $\mathcal{B}_{i,j} = \langle 1, \bb_{i,j}, \bc_{i,j}, \bk_{i,j} \rangle$
of the $q$-boson algebras for $1 \le i \le j < n$ obeying (\ref{akn}) within each $\mathcal{B}_{i,j}$ and $[\mathcal{B}_{i,j}, \mathcal{B}_{i',j'}]=0$
if $(i,j) \neq (i',j')$. We take them so that $Z_{\alpha_1,\ldots, \alpha_{n}}(\zeta) \in \bigotimes_{1 \le i \le j < n}\mathcal{B}_{i,j}$ and~(\ref{ain}) reads
\begin{gather*}
Y_n(\zeta)= \sum_{l_1,\ldots, l_{n-1} \in \Z_{\ge 0}} g_{l_1,\ldots, l_{n-1}}(\zeta) \bigl(\bk_{1,n-2}^{l^+_1}\bc_{1,n-2}^{l_1} \cdots
\bk_{n-2,n-2}^{l^+_{n-2}}\bc_{n-2,n-2}^{l_{n-2}}\bigr) \bigl(\bb_{1,n-1}^{l_1}\cdots \bb_{n-1,n-1}^{l_{n-1}}\bigr),
\end{gather*}
where $l^+_j = l_{j+1}+ \cdots + l_{n-1}$. It corresponds to labeling the components in $\mathcal{B}^{\otimes n(n-1)/2}$ as
\begin{gather}\label{hsi}
(1,1), (1,2), (2,2), (1,3), (2,3), (3,3), \ldots, (1,n-1),(2,n-1), \ldots, (n-1,n-1).
\end{gather}
Def\/ine the following elements in $\bigotimes_{1 \le i \le j < n}\mathcal{B}_{i,j}$ (actually in a certain completion of it):
\begin{gather*}
Y_j(\zeta) = V_{j-1}(1)V_{j-1}(\zeta)^{-1}= V_{j-1}(\zeta)^{-1}V_{j-1}(1),\\
V_j(\zeta) = (\zeta^{-1}A_{1,j})_\infty (\zeta^{-1}A_{2,j})_\infty \cdots (\zeta^{-1}A_{j,j})_\infty,\\
A_{i,j}= \bk_{1,j-1}\bk_{2,j-1}\cdots \bk_{i-1,j-1}\bc_{i,j-1}\bb_{i,j}, \qquad \bc_{j,j-1}=1.
\end{gather*}
In particular, $A_{1,n-1}=\bc_{1,n-2}\bb_{1,n-1}$ and $A_{n-1,n-1} = \bk_{1,n-2} \cdots \bk_{n-2,n-2}\bb_{n-1,n-1}$. As for the right equality in the f\/irst line, see Remark~\ref{askg}.

\begin{Theorem}[\cite{KO2}]\label{th:nzm}
The representation of the ZF algebra \eqref{mrn} in $\bigotimes_{1 \le i \le j < n}\mathcal{B}_{i,j}$ given in Theorem~{\rm \ref{th:main}} and \eqref{ini}, \eqref{mdk} is expressed as follows:
\begin{gather*}
Z_{\alpha_1,\ldots, \alpha_{n}}(\zeta) = Z_{0^{n}}(\zeta)\bk_{1,n-1}^{\alpha^+_1}\bc_{1,n-1}^{\alpha_1}\cdots \bk_{n-1,n-1}^{\alpha^+_{n-1}}\bc_{n-1,n-1}^{\alpha_{n-1}}, \qquad \alpha^+_i = \alpha_{i+1}+\cdots + \alpha_{n},\\
Z_{0^{n}}(\zeta)= Y_2(\zeta)Y_3(\zeta) \cdots Y_{n}(\zeta).
\end{gather*}
\end{Theorem}

The cases $n= 2,3$ reproduce (\ref{knh}) under the identif\/ication ${\bf x}_{1,1}={\bf x} \otimes 1 \otimes 1$, ${\bf x}_{1,2} = 1 \otimes {\bf x} \otimes 1$, ${\bf x}_{2,2} = 1 \otimes 1 \otimes {\bf x}$ in accordance with~(\ref{hsi}). The recursive construction (\ref{zgz}) with respect to rank may be viewed as reminiscent of nested Bethe ansatz. A similar structure has also been observed in multispecies ASEP \cite{CDW, PEM}.

It is not hard to show that the substitution of the formulas in Theorem~\ref{th:nzm} and~(\ref{aim}) into the matrix product formula~(\ref{mst}) of $\mathbb{P}(\sigma_1, \ldots, \sigma_L)$ leads to the {\em convergent} trace provided that the conf\/iguration $(\sigma_1, \ldots, \sigma_L)$ belongs to a~basic sector explained in Section~\ref{sec4.1}.

\begin{Remark}\label{askg}
As a corollary of the ZF algebra (\ref{msk}) and $\mathscr{S}(\lambda, \mu)_{\gamma,\delta}^{0,0} = \theta(\gamma=\delta = 0)$ one can derive
\begin{gather*}
[Z_0(\mu), Z_0(\lambda)]=0, \qquad [V_m(\mu), V_m(\lambda)]=0,\qquad 1 \le m \le n-1.
\end{gather*}
\end{Remark}

Let us comment on the implication of the sum-to-unity property (\ref{sum2}) to the ZF algebra~(\ref{msk}). Using the weight conservation it implies
\begin{gather}\label{kand}
[A(\mu|w), A(\lambda|w)]=0\qquad \text{for} \quad A(\lambda|w) = \sum_{\alpha}X_\alpha(\lambda)w^\alpha,
\end{gather}
where $w^\alpha = w_1^{\alpha_1}\cdots w_n^{\alpha_n}$.

\begin{Example}\label{kany}
The generating function in (\ref{kand}) for $n=2$ and $w=(w_1,w_2)=(x,y)$ can be computed as
\begin{gather*}
A(\lambda|w) =\frac{(\bb)_\infty}{(\lambda^{-1}\bb)_\infty} \sum_{l,m\ge 0}\frac{x^m y^l \lambda^{-l-m}(\lambda)_{l+m}}{(q)_l(q)_m} \bk^l \bc^m\\
\hphantom{A(\lambda|w)}{} =\frac{(\bb)_\infty}{\big(\lambda^{-1}\bb\big)_\infty}\sum_{m\ge 0}\frac{(x\lambda^{-1})^{m}(\lambda)_m}{(q)_m}
\frac{(q^my\bk)_\infty}{(y\lambda^{-1}\bk)_\infty}\bc^m\\
\hphantom{A(\lambda|w)}{}
= \frac{(\bb)_\infty}{(\lambda^{-1}\bb)_\infty} \Gamma\big(x\lambda^{-1}, y\lambda^{-1}\big)^{-1}\Gamma(x,y),\qquad \Gamma(x,y) = (x \bc)_\infty(y \bk)_\infty.
\end{gather*}
It plays an important role to investigate the stationary states in `grand canonical ensemble' picture (cf.~\cite{EH}).
\end{Example}

In the matrix product formula (\ref{mst}), multiply $w^{\sigma_1+\cdots + \sigma_L}$ and take the sum over the states $(\sigma_1,\dots, \sigma_L) \in (\Z^n_{\ge 0})^L$ belonging to all the basic sectors. It yields the series
\begin{gather*}
\operatorname{Tr}(A(\mu_1|w)\cdots A(\mu_L|w))' = \sum_{\alpha_1+\cdots + \alpha_L \in \Z_{\ge 1}^n}w^{\alpha_1+\cdots + \alpha_L}
\operatorname{Tr}(X_{\alpha_1}(\mu_1) \cdots X_{\alpha_L}(\mu_L))\\
\hphantom{\operatorname{Tr}(A(\mu_1|w)\cdots A(\mu_L|w))'}{} = \sum_{k \in \Z_{\ge 1}^n} w^k G_k(\mu_1,\ldots, \mu_L; q),
\end{gather*}
where the prime drops the terms corresponding to $\alpha_1+\cdots + \alpha_L \in \Z_{\ge 0}^n\setminus \Z_{\ge 1}^n$ to take into account the basic sectors only. We have introduced the function
\begin{gather*}
G_k(\mu_1,\ldots, \mu_L; q) = \sum_{\alpha_1+\cdots + \alpha_L=k} \operatorname{Tr}(X_{\alpha_1}(\mu_1) \cdots X_{\alpha_L}(\mu_L)).
\end{gather*}
This is the normalization constant of the stationary state in the basic sector $\mathcal{B}(k)$ (\ref{hrk2}). The commutativity in (\ref{kand}) implies that $G_k(\mu_1,\ldots, \mu_L; q)$ is symmetric function in $\mu_1,\ldots, \mu_L$. Moreover from (\ref{mgd}), (\ref{aim}), Theorem \ref{th:nzm} and the expansion formula $\frac{(zw)_\infty}{(z)_\infty} = \sum\limits_{j \ge 0}\frac{(w)_j}{(q)_j}z^j$, it follows that $G_k(\mu_1,\ldots, \mu_L; q)$ is a {\em polynomial} in $\mu^{-1}_i$'s. It is rational in $q$ because the trace is evaluated as $\operatorname{Tr}(\bk^r) = (1-q^r)^{-1}$. Detailed study of $G_k(\mu_1,\ldots, \mu_L; q)$ and the related Conjecture~\ref{c:aim} is a~future problem\footnote{We expect that the condition $k \in \Z_{\ge 1}^n$ can slightly be relaxed in view of the convergence of $\operatorname{Tr}(X_{\alpha_1}(\mu_1) \cdots X_{\alpha_L}(\mu_L))$ in the non-basic sector $(k_1,k_2)$ with $k_1=0$, $k_2\ge 1$ for $n=2$. See~(\ref{Xw}).}. Similar observations in such a direction have led to matrix product formulas for Macdonald polynomials and their generalizations for some other integrable lattice models. See for example \cite{BP,CDW,CGDW,GDW,MS}.

\section[$R$ matrices of generalized quantum group]{$\boldsymbol{R}$ matrices of generalized quantum group}\label{sec6}
We have seen that the factorization (\ref{fac}) led to signif\/icant consequences in previous sections. Here we generalize it to a part of $2^{n+1}$ quantum $R$ matrices labeled with $(\ep_1,\ldots,\ep_{n+1}) \in \{0,1\}^{n+1}$. The one treated so far corresponds to the choice $(\ep_1,\ldots,\ep_{n+1}) =(0,\ldots, 0)$. These $R$ matrices have been obtained from the special solutions to the tetrahedron equation by a certain reduction~\cite{KOS}. The underlying algebra has been identif\/ied with a generalized quantum group. We shall present these results with a brief background based on~\cite{KOS}.

\subsection{Def\/inition of ${\mathcal U}_A$} \label{subsec:U_A}
In this subsection we assume that $n$ is a positive integer. Let $(\ep_1,\ldots,\ep_{n+1})$ be a sequence of $0$ or $1$. In what follows the indices $i,j$ are understood to be elements in $\Z_{n+1}$. We write $i\equiv j$ to mean $i=j$ in $\Z_{n+1}$. For $i,j=0,1,\ldots,n\in\Z_{n+1}$, set
\begin{gather*}
q_i=\begin{cases}
q, &\ep_i=0, \\
-q^{-1}, &\ep_i=1,
\end{cases}
\qquad
D_{ij}=D_{ji}=\begin{cases}
q_iq_{i+1},& j\equiv i, \\
q_i^{-1},& j\equiv i-1, \\
q_{i+1}^{-1},& j\equiv i+1, \\
1, & \text{otherwise}.
\end{cases}
\end{gather*}
Let $\U_A=\U_A(\ep_1,\ldots,\ep_{n+1})$ be a $\Q(q)$-algebra generated by $e_i$,$ f_i$, $k^{\pm 1}_i$ $(i\in \Z_{n+1})$ obeying the following relations. (We use the notation $[u]=(q^u-q^{-u})/(q-q^{-1})$.)
\begin{gather}
 k_i k^{-1} _i=k^{-1} _i k_i=1, \qquad k_ik_j=k_jk_i, \qquad k_i e_jk_i^{-1}=D_{ij}e_j, \qquad k_i f_jk^{-1}_i=D^{-1}_{ij}f_j,\label{kkd}\\
[e_i ,f_j]= \delta_{ij} \frac{k_i-k^{-1} _i}{q-q^{-1}},\label{efk}\\
e^2 _i =f^2 _i =0 \qquad\text{if} \quad \ep_i \neq \ep_{i+1}, \label{w:relA4}\\
[e_i ,e_j]=[f_i ,f_j]=0 \qquad\text{if} \quad j\not\equiv i,i\pm1, \\
e^2 _i e_j - (-1)^{\ep_i}[2]e_i e_j e_i +e_j e^2 _i=(e \rightarrow f)=0 \qquad\text{if} \quad \ep_i =\ep_{i+1},j\equiv i\pm1, \\
e_{i}e_{i-1}e_{i}e_{i+1}+(-1)^{\ep_i}[2]e_{i}e_{i-1}e_{i+1}e_{i} -e_{i}e_{i+1}e_{i}e_{i-1} \nonumber\\
\qquad {} -e_{i-1}e_{i}e_{i+1}e_{i}+e_{i+1}e_{i}e_{i-1}e_{i} =(e \rightarrow f)=0 \qquad\text{if} \quad \ep_i \neq \ep_{i+1}. \label{w:relA7}
\end{gather}

The algebra $\U_A$ with the relations (\ref{kkd}) and (\ref{efk}) was introduced for $n \ge 1$ in~\cite{KOS} as a~symmetry algebra characterizing solutions to the YBE obtained by the 2D reduction procedure to be explained in Section~\ref{sec6.2} from the tetrahedron equation \cite{Zam80}. Based on the observation in \cite[Section~3.3]{KOS}, the relations (\ref{w:relA4})--(\ref{w:relA7}) were supplemented when $n \ge 2$ \cite{Mach} by showing that the subalgebra generated by $e_i$, $f_i$, $k_i$ for $i=1,\ldots,n$ is isomorphic, up to adding simple generators, to the quantized universal enveloping super algebra of type~$A$ \cite{Yamane}, where they correspond to a~$q$-analogue of the Serre relations. The forthcoming construction~(\ref{tgm}) and all the subsequent claims
are valid for $n \ge 1$. $\U_A$ is a Hopf algebra with coproduct $\Delta$ given by
\begin{gather}
\Delta k^{\pm 1}_i = k^{\pm 1}_i\otimes k^{\pm 1}_i,\qquad \Delta e_i = 1\otimes e_i + e_i \otimes k_i,\qquad \Delta f_i = f_i\otimes 1 + k^{-1}_i\otimes f_i. \label{w:Delta}
\end{gather}
For the counit and the antipode, see \cite[equation~(3.4)]{KOS}. To present a representation of~$\mathcal{U}_{A}$, we introduce the following vector spaces:
\begin{gather}
 F=W^{(0)} = \bigoplus_{m\ge0}\C|m\rangle^{(0)}, \qquad V=W^{(1)}=\C|0\rangle^{(1)}\oplus\C|1\rangle^{(1)}, \label{fv}\\
\mathcal{W} = W^{(\ep_1)} \ot \cdots \ot W^{(\ep_{n+1})} = \bigoplus_{\alpha_1,\dots,\alpha_{n+1}}\C |\alpha_1,\dots,\alpha_{n+1}\rangle ,\label{w:quantsp}\\
 |\alpha_1,\dots,\alpha_{n+1} \rangle= |\alpha_1 \rangle^{(\epsilon_1)} \ot \cdots \ot |\alpha_{n+1}\rangle^{(\epsilon_{n+1})},\nonumber\\
V_{l}= \bigoplus_{\alpha, |\alpha|=l}\C|\alpha \rangle \subset \mathcal{W}.\label{vll}
\end{gather}
Note that the range of the indices $\alpha_i$ are to be understood as $\Z_{\ge 0}$ or $ \{ 0, 1\}$ according to $\epsilon_i = 0$ or~$1$, respectively. In (\ref{vll}) we have written $|\alpha_1,\cdots,\alpha_{n+1}\rangle $ simply as $|\alpha \rangle$.

\begin{Proposition}[\cite{KOS}]\label{w:repU_A}
Let $x$ be a parameter. The map $\pi_{x}^{l} \colon \mathcal{U}_{A} \rightarrow \operatorname{End}(V_{l})$ defined by
\begin{gather}
\pi_{x}^{l}(e_{i}) |\alpha \rangle = x^{\delta_{i,0}}[\alpha_{i}]|\alpha - {\bf e}_i + {\bf e}_{i+1} \rangle, \nonumber \\
\pi_{x}^{l}(f_{i}) |\alpha \rangle =x^{-\delta_{i,0}}[\alpha_{i+1}]|\alpha + {\bf e}_i - {\bf e}_{i+1} \rangle,\nonumber \\
\pi_{x}^{l}(k_{i})|\alpha \rangle=(q_{i})^{-\alpha_{i}}(q_{i+1})^{\alpha_{i+1}}|\alpha \rangle \label{w:rep}
\end{gather}
for $0\le i\le n$ gives an irreducible representation when $0 \le l \le n+1$ if $\epsilon_{1} = \cdots = \epsilon_{n+1}=1$, and $l \in \Z_{\ge 0}$ otherwise. The index $i$ should be understood mod $n+1$. In the right hand side of~\eqref{w:rep}, vectors $|\alpha'\rangle = |\alpha_{1}', \dots,\alpha_{n+1}' \rangle$
are to be understood as zero unless $\alpha_{i} \in \{0,1\}\, (\epsilon_i=1)$ and $\alpha_{i} \in \Z_{\ge 0}$ $(\epsilon_i=0)$ for all $1 \le i \le n+1$.
\end{Proposition}

Consider the following linear equation on $R \in \operatorname{End}(V_l \otimes V_m)$:
\begin{gather}\label{wsy}
\big(\pi^l_x \otimes \pi^m_y\big)\Delta^{\rm op}(g) R = R\big(\pi^l_x \otimes \pi^m_y\big)\Delta(g)\qquad \forall\, g \in \mathcal{U}_{A},
\end{gather}
where $\Delta^{\rm op}(g) = P \circ \Delta \circ P$ with $P(u \otimes v ) = v \otimes u$. The dimension of the solution space to this equation is at most one
if the $\mathcal{U}_{A}$-module $V_l \otimes V_m$ is irreducible.

\begin{Conjecture}[\cite{KOS}] \label{yssi}
The $\mathcal{U}_A$-module $V_{l} \ot V_{m}$ is irreducible for any choice $(\ep_1,\ldots,\ep_{n+1}) \in \{0,1\}^{n+1}$.
\end{Conjecture}

\begin{Theorem}[\cite{KOS}]\label{th:true}
Conjecture {\rm \ref{yssi}} is true for $(\ep_1,\ldots,\ep_{n+1}) =(1^\kappa,0^{n+1-\kappa})$ with $0 \le \kappa \le n+1$.
\end{Theorem}

Suppose $(\ep_1,\ldots,\ep_{n+1}) =(1^\kappa,0^{n+1-\kappa})$ with $0 \le \kappa \le n+1$. Then a little inspection on the representation~(\ref{w:rep}) tells that the solution to (\ref{wsy}) depends on $x$, $y$ as $R=R(z)$ with $z=x/y$. It will be denoted by $R(z)=R^{l,m}(z)$ and referred to as {\em quantum $R$ matrix} up to overall normalization. To f\/ix the normalization, introduce
\begin{gather}\label{hrm}
{\bf e}_i=(\overset{i-1}{\overbrace{0,\ldots,0}}, 1,\overset{n+1-i}{\overbrace{0,\ldots,0}}) \in \Z^{n+1},
\qquad {\bf e}_{>m}={\bf e}_{m+1}+\cdots+{\bf e}_{n+1}, \quad 1 \le m \le n+1.
\end{gather}
If $\ep_1 = \cdots = \ep_{n+1} =1$, we normalize $R^{l,m}(z)$ $(0 \le l,m \le n+1)$ as
\begin{gather}\label{yuj}
R^{l,m}(z)(|{\bf e}_{>n+1-l} \rangle \ot |{\bf e}_{>n+1-m} \rangle ) = |{\bf e}_{>n+1-l} \rangle \ot |{\bf e}_{>n+1-m} \rangle.
\end{gather}
If $\ep_1 \cdots \ep_{n+1} =0$, pick any $i$ such that $\ep_{i}=0$ and normalize it as
\begin{gather}\label{yuk}
R^{l,m}(z) (| l{\bf e}_{i} \rangle \ot | m{\bf e}_{i} \rangle ) = | l{\bf e}_{i} \rangle \ot | m{\bf e}_{i} \rangle.
\end{gather}
In view of Conjecture \ref{yssi}, we expect that the quantum $R$ matrix is characterized similarly for any $(\ep_1,\ldots,\ep_{n+1}) \in \{0,1\}^{n+1}$.

\subsection[Construction of $R$ matrix]{Construction of $\boldsymbol{R}$ matrix}\label{sec6.2}
We brief\/ly review how we constructed in \cite{KOS} solutions depending on $(\ep_1,\ldots,\ep_{n+1}) \in \{0,1\}^{n+1}$ to the YBE from the tetrahedron equation. Def\/ine the 3D $R$ operator $\Rm \in \operatorname{End}\big(W^{(0)}\otimes W^{(0)}\otimes W^{(0)}\big)$ by
\begin{gather*}
\Rm \big(|i\rangle ^{(0)}\otimes|j\rangle ^{(0)}\otimes|k\rangle ^{(0)}\big) = \sum_{a,b,c \in \Z_{\ge 0}} \Rm^{a,b,c}_{i,j,k}
|a\rangle ^{(0)}\otimes|b\rangle ^{(0)}\otimes|c\rangle ^{(0)},
\end{gather*}
where $\Rm^{a,b,c}_{i,j,k}$ is the one in (\ref{Rint}). Similarly def\/ine the 3D $L$ operator $\mathscr{L} \in \operatorname{End}\big(W^{(1)} \otimes W^{(1)}\otimes W^{(0)}\big)$~\cite{BS} by
\begin{gather*}
 \mathscr{L} (|i\rangle ^{(1)}\otimes|j\rangle ^{(1)}\otimes|k\rangle ^{(0)}) = \sum_{a,b \in \{0,1\}, \,c \in \Z_{\ge 0}} \mathscr{L} ^{a,b,c}_{i,j,k}
|a\rangle ^{(1)}\otimes|b\rangle ^{(1)}\otimes|c\rangle ^{(0)},\\
{\mathscr L}^{0, 0, j}_{0, 0, m} ={\mathscr L}^{1, 1, j}_{1, 1, m}=\delta^j_m,\qquad {\mathscr L}^{0, 1, j}_{0, 1, m}= -\delta^j_mq^{m+1},\qquad
{\mathscr L}^{1, 0, j}_{1, 0, m}=\delta^j_mq^m,\\
{\mathscr L}^{0, 1, j}_{1, 0, m}=\delta^j_{m-1}\big(1-q^{2m}\big),\qquad {\mathscr L}^{1, 0, j}_{0, 1, m}=\delta^j_{m+1}.
\end{gather*}
It may be viewed as a six-vertex model with $q$-boson valued Boltzmann weights in the third component.

Assign a solid arrow to $F$ and a dotted arrow to $V$, and depict the matrix elements of 3D $R$ and 3D $L$ as
\begin{gather*}
\includegraphics{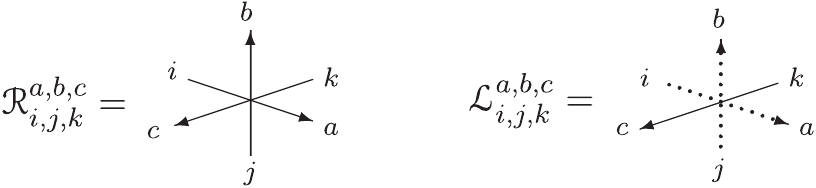}
\end{gather*}

To treat $\Rm$ and $\mathscr{L}$ on an equal footing we set $\mathscr{M}^{(0)}=\Rm$ and $\mathscr{M}^{(1)}=\mathscr{L}$ so that $\mathscr{M}^{(\ep)} \in
\operatorname{End} \big(W^{(\ep)} \ot W^{(\ep)} \ot F\big)$. They satisfy the following type of tetrahedron equation
\cite{BS,KV,KOS}\footnote{$\mathscr{M}^{(0)}$, $\mathscr{M}^{(1)}$ were denoted by $\mathscr{S}^{(0)}$, $\mathscr{S}^{(1)}$ in \cite{KOS}.}.
\begin{gather}
\mathscr{M}^{(\ep)}_{1,2,4}\mathscr{M}_{1,3,5}^{(\ep)} \mathscr{M}_{2,3,6}^{(\ep)}\mathscr{R}_{4,5,6} =\mathscr{R}_{4,5,6}\mathscr{M}_{2,3,6}^{(\ep)}
\mathscr{M}_{1,3,5}^{(\ep)}\mathscr{M}_{1,2,4}^{(\ep)}.\label{w:tetra}
\end{gather}
This is an equality in $\operatorname{End}\big(W^{(\ep)} \ot W^{(\ep)} \ot W^{(\ep)} \ot F \ot F \ot F\big)$. Subscripts of $\mathscr{M}^{(\ep)}_{i,j,k}$
or $\mathscr{R}_{i,j,k}$ signify that they act on the $i,j$ and $k$-th components of $W^{(\ep)} \ot W^{(\ep)} \ot W^{(\ep)} \ot F \ot F \ot F$, and
do as the identity on the other. The tetrahedron equation~(\ref{w:tetra}) with $\epsilon=1$ is depicted as follows:
\begin{gather*}
\includegraphics{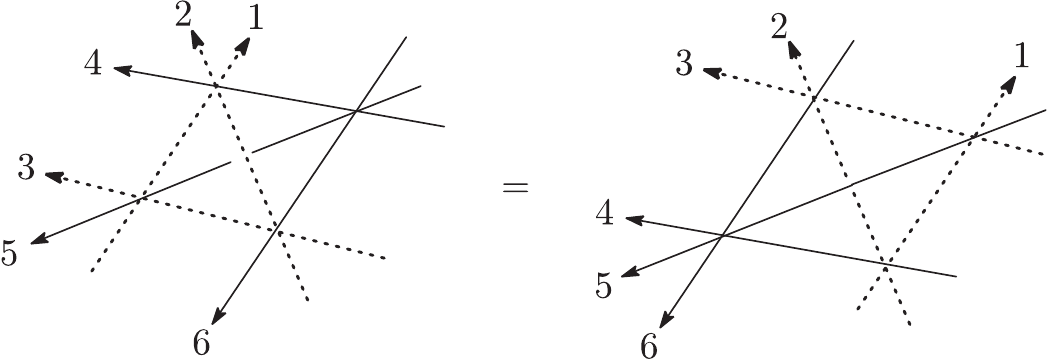}
\end{gather*}

The equation (\ref{w:tetra}) with $\epsilon=0$ is expressed similarly by replacing all the dotted arrows by solid ones.

Regarding \eqref{w:tetra} as a one-layer relation, we extend it to the $(n+1)$-layer version. Let $\overset{a_i}{W^{(\ep_i)}}$, $\overset{b_i}{W^{(\ep_i)}}$, $\overset{c_i}{W^{(\ep_i)}}$ be copies of $W^{(\ep_i)}$, where~$a_i$,~$b_i$ and~$c_i$ $(i=1,\dots,n+1)$ are just distinct labels. Repeated use of \eqref{w:tetra} $(n+1)$ times leads to
\begin{gather}
\big(\mathscr{M}^{(\ep_1)}_{a_1,b_1,4} \mathscr{M}_{a_1,c_1,5}^{(\ep_1)} \mathscr{M}_{b_1,c_1,6}^{(\ep_1)}\big) \cdots
\big(\mathscr{M}^{(\ep_{n+1})}_{a_{n+1},b_{n+1},4} \mathscr{M}_{a_{n+1},c_{n+1},5}^{(\ep_{n+1})} \mathscr{M}_{b_{n+1},c_{n+1},6}^{(\ep_{n+1})}\big)
\mathscr{R}_{4,5,6} \nonumber \\
\qquad{} =\mathscr{R}_{4,5,6} \big(\mathscr{M}_{b_1,c_1,6}^{(\ep_1)} \mathscr{M}_{a_1,c_1,5}^{(\ep_1)}\mathscr{M}^{(\ep_1)}_{a_1,b_1,4}\big)
\cdots \big(\mathscr{M}_{b_{n+1},c_{n+1},6}^{(\ep_{n+1})} \mathscr{M}_{a_{n+1},c_{n+1},5}^{(\ep_{n+1})} \mathscr{M}^{(\ep_{n+1})}_{a_{n+1},b_{n+1},4}\big). \label{w:n-layer}
\end{gather}
This is an equality in $\operatorname{End} \big(\overset{a}{\mathcal{W}} \ot \overset{b}{\mathcal{W}} \ot \overset{c}{\mathcal{W}} \ot \overset{4}{F} \ot \overset{5}{F} \ot \overset{6}{F}\big)$, where $a=(a_1,\dots,a_{n+1})$ is the array of labels and $\overset{a}{\mathcal{W}} =\overset{a_1}{W^{(\ep_1)}} \ot \cdots \ot \overset{a_{n+1}}{W^{(\ep_{n+1})}}$. The notation $\overset{b}{\mathcal{W}}$ and $\overset{c}{\mathcal{W}}$ should be understood similarly. They are just copies of $\mathcal{W}$ def\/ined in \eqref{w:quantsp}. One can reduce~\eqref{w:n-layer} to the YBE by evaluating the auxiliary space $\overset{4}{F} \ot \overset{5}{F} \ot \overset{6}{F}$ away appropriately. A natural way is to take trace of~\eqref{w:n-layer} over the auxiliary space after multiplying it with
$x^{\rm{\bf h}_4}(xy)^{\rm{\bf h}_5}y^{\rm{\bf h}_6}$ from the left and $\mathscr{R}_{4,5,6}^{-1}$ from the right\footnote{See around \cite[equation~(2.5)]{KOS} for the def\/inition of ${\bf h}_i$.}. It results in the YBE
\begin{gather}\label{yyb}
R_{a,b}(x)R_{a,c}(xy)R_{b,c}(y)=R_{b,c}(y)R_{a,c}(xy)R_{a,b}(x) \in \operatorname{End}\big(\overset{a}{\mathcal{W}} \ot \overset{b}{\mathcal{W}} \ot \overset{c}{\mathcal{W}}\big)
\end{gather}
for the $R$ matrix obtained as
\begin{gather}\label{tgm}
R_{a,b}(z) =\rho(z)\operatorname{Tr}_{3}\bigl(z^{\rm{\bf h}_3} \mathscr{M}^{(\ep_1)}_{a_1,b_1,3} \cdots \mathscr{M}^{(\ep_{n+1})}_{a_{n+1},b_{n+1},3}\bigr)
\in \operatorname{End}\big(\overset{a}{\mathcal{W}} \ot \overset{b}{\mathcal{W}}\big),
\end{gather}
where the scalar $\rho(z)$ is inserted to control the normalization. The trace are taken with respect to the auxiliary Fock space $F=\overset{3}{F}$ signif\/ied by~$3$. Pictorially the matrix element~(\ref{w:matrixele}) of~(\ref{tgm}) is expressed as follows:
\begin{gather*}
\includegraphics{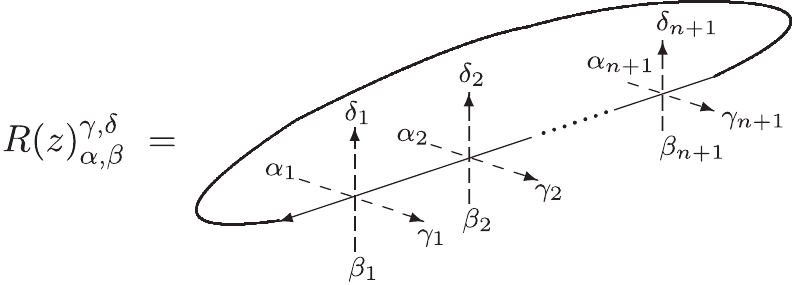}
\end{gather*}
Here the broken arrows designate either $F$ or $V$ (\ref{fv}) according to $\epsilon_i = 0$ or $1$, and the winding arrow does $\overset{3}{F}$ over which the trace is taken. In short (\ref{tgm}) is a matrix product construction of quantum $R$ matrix $R(z)$ by operators satisfying the tetrahedron equation. The formula (\ref{rel}) is just the concrete form of~(\ref{tgm}) for $(\epsilon_1,\ldots, \epsilon_{n+1}) = (0,\ldots, 0)$.

It is not known if this type of construction extends much beyond the generalized quantum group $\mathcal{U}_A$. See \cite[Section~2.8]{KOS} for the list of the known results. However, the formula (\ref{tgm}) is often more ef\/f\/icient than the fusion procedure practically. It also reveals a hidden 3D structure in a class of~$R$ matrices~\cite{BS} and has led to another application to the multispecies totally asymmetric simple exclusion and zero range processes
when $\forall\, \epsilon_i=1$ and $\forall\, \epsilon_i=0$ \cite{KMO0,KMO2}. Except for the two cases however, these $R$ matrices do {\em not} satisfy the sum-to-unity in general\footnote{This is partly because \cite[Lemma~5]{KMMO} becomes trivial for $s \ge 3$ invalidating the argument similar to the proof of \cite[Theorem~6]{KMMO}.} and we have not found an application to stochastic systems.

Def\/ine the matrix elements of $R(z)$ by
\begin{gather}
R(z)(|\alpha\rangle \ot |\beta \rangle ) = \sum_{\gamma,\delta}R(z)_{\alpha,\beta}^{\gamma,\delta}|\gamma \rangle \ot |\delta \rangle,\label{w:matrixele}
\end{gather}
where $|\alpha\rangle, \ldots, |\delta\rangle \in \mathcal{W}$ (\ref{w:quantsp}).
It satisf\/ies
\begin{gather*}
R(z)_{\alpha,\beta}^{\gamma,\delta} = 0 \qquad {\rm unless} \quad \gamma+\delta =\alpha+\beta \quad {\rm and} \quad |\gamma|=|\alpha|, \quad |\delta|=|\beta|,
\end{gather*}
which implies that $R(z)$ decomposes into matrices acting on f\/inite-dimensional vector spaces:
\begin{gather*}
R(z)=\bigoplus_{l,m \ge 0} R^{l,m}(z),\qquad R^{l,m}(z)\in \operatorname{End}(V_{l}\ot V_m),
\end{gather*}
where the former sum ranges over $0 \le l,m \le n+1$ if $\epsilon_1 = \cdots = \epsilon_{n+1} =1$ and $l,m \in \Z_{\ge 0}$ otherwise\footnote{The notation $V_l$ matches (\ref{BV}) when $\forall\, \epsilon_i=0$. It was denoted by~$\mathcal{W}_l$ in~\cite{KOS} although.}. The normalization (\ref{yuj}) and (\ref{yuk}) is achieved by choosing $\rho(z)=(-q)^{-\max (m-l,0)}(1-q^{|l-m|}z)$ and $\rho(z) =\frac{z^{-m}(q^{l-m}z;q^{2})_{m+1}}{(q^{l-m+2}z^{-1};q^{2})_{m}}$ in~(\ref{tgm}), respectively \cite[Section~2.6]{KOS}. When $l=m$, we have
\begin{gather*}
R^{m,m}(1)(|\alpha \rangle \ot |\beta \rangle) =|\beta \rangle \ot |\alpha \rangle.
\end{gather*}

\begin{Proposition}[\cite{KOS}]\label{pr:kos51}
For arbitrary sequence $(\epsilon_1,\ldots, \epsilon_{n+1}) \in \{0,1\}^{n+1}$, the matrix $R^{l,m}(z=x/y)$ satisfies~\eqref{wsy}.
\end{Proposition}

From Theorem \ref{th:true} and Proposition \ref{pr:kos51} it follows that $R^{l,m}(z)$ associated with the sequence $(\ep_1,\ldots,\ep_{n+1}) =(1^\kappa,0^{n+1-\kappa})$ with $0 \le \kappa \le n+1$ is indeed the quantum $R$ matrix of $\mathcal{U}_A$.

\begin{Example}
Consider $\mathcal{U}_{A}(1,0)$. For $l,m \ge 1$, one has $V_{m} = \C|0,m\rangle \oplus \C|1,m-1 \rangle \subset \mathcal{W} = V \ot F$ and similarly for $V_{l}$. The action of $R(z)$ on $V_{l} \ot V_{m}$ is given by
\begin{gather*}
R(z)(|0,l \rangle \ot |0,m \rangle) = |0,l \rangle \ot |0,m \rangle, \\
R(z)(|1,l-1 \rangle \ot |0,m \rangle) =\frac{1-q^{2m}}{z-q^{l+m}}|0,l \rangle \ot |1,m-1 \rangle+ \frac{q^{m}z-q^{l}}{z-q^{l+m}}|1,l-1 \rangle \ot |0,m \rangle, \\
R(z)(|0,l \rangle \ot |1,m-1 \rangle)=\frac{q^{l}z-q^{m}}{z-q^{l+m}}|0,l \rangle \ot |1,m-1 \rangle + \frac{(1-q^{2l})z}{z-q^{l+m}}
|1,l-1 \rangle \ot |0,m \rangle, \\
R(z)(|1,l-1 \rangle \ot |1,m-1 \rangle) =\frac{1-q^{l+m}z}{z-q^{l+m}}|1,l-1 \rangle \ot |1,m-1 \rangle.
\end{gather*}
\end{Example}

\subsection[Special value of $R$ matrix]{Special value of $\boldsymbol{R}$ matrix}
In this subsection we wish to obtain an explicit form of $R(z)=R^{l,m}(z)$ at $z=q^{l-m}$. We are going to show the following theorem.
\begin{Theorem} \label{th:specialized R}
When $l\le m$ and $(\ep_1,\ldots,\ep_{n+1}) =(1^\kappa,0^{n+1-\kappa})$ with $0 \le \kappa \le n+1$, the following formula is valid:
\begin{gather}
R\big(z=q^{l-m}\big)_{\alpha,\beta}^{\gamma,\delta} =\delta_{\alpha+\beta}^{\gamma+\delta}q^{\psi+l(l-m)\theta(\kappa=n+1)}
\binom{m}{l}_{q^{2}}^{\theta(\kappa=n+1)-1} \prod_{i=1}^{n+1}\binom{\beta_i}{\gamma_i}_{q^{2}}, \label{w:defQ} \\
\psi=\psi_{\alpha,\beta}^{\gamma,\delta} =\sum_{1\le i < j \le n+1}\alpha_{i}(\beta_{j} - \gamma_{j})
+ \sum_{1\le i < j \le n+1}(\beta_{i} - \gamma_{i})\gamma_{j}. \nonumber
\end{gather}
\end{Theorem}
Set $\hat{i}= {\bf e}_i-{\bf e}_{i+1}$ for $i \in \Z_{n+1}$. It is straightforward to check
\begin{Lemma}\label{w:lempsi}
\begin{gather*}
\psi_{\alpha,\beta}^{\gamma-\hat{i},\delta} -\psi_{\alpha,\beta}^{\gamma,\delta-\hat{i}} = \gamma_{i+1} -\alpha_{i}+\beta_{i} -\gamma_{i} + 1 + (l-m)\delta_{i,0}, \\
\psi_{\alpha+\hat{i},\beta}^{\gamma,\delta}-\psi_{\alpha,\beta}^{\gamma,\delta-\hat{i}}=\beta_{i+1}-\gamma_{i+1}+ (l-m)\delta_{i,0},\qquad
\psi_{\alpha,\beta+\hat{i}}^{\gamma,\delta}-\psi_{\alpha,\beta}^{\gamma,\delta-\hat{i}}=\gamma_{i+1} - \alpha_{i} .
\end{gather*}
\end{Lemma}

\begin{Proposition}\label{w:proposition2}
Denote the r.h.s.\ of \eqref{w:defQ} by $X^{\gamma,\delta}_{\alpha,\beta}$ and set $X(|\alpha \rangle \ot |\beta \rangle ) =\sum\limits_{\gamma,\delta}X^{\gamma,\delta}_{\alpha,\beta} |\gamma \rangle\ot |\delta \rangle$. Suppose $l\le m$. Then for any $x$, $y$ such that $x/y=q^{l-m}$ and $(\ep_1,\ldots,\ep_{n+1}) \in \{0,1\}^{n+1}$, we have
\begin{gather}\label{w:comrel}
\big(\pi^l_x \otimes \pi^m_y\big)\Delta^{\rm op}(g) X = X\big(\pi^l_x \otimes \pi^m_y\big)\Delta(g)
\qquad \forall\, g \in \mathcal{U}_{A}(\epsilon_{1},\dots,\epsilon_{n+1}).
\end{gather}
\end{Proposition}

\begin{proof}
It suf\/f\/ices to show it for $g=k_i, e_i, f_i$ in the four cases $(\epsilon_{i},\epsilon_{i+1})=(1,1),(1,0),(0,1)$, $(0,0)$. The case $(\epsilon_{i},\epsilon_{i+1})=(0,0)$ was done in~\cite{KMMO}. Here we treat $(\epsilon_{i},\epsilon_{i+1})=(0,1)$ case. The proof for $(\epsilon_{i},\epsilon_{i+1})=(1,1),(1,0)$ are similar. The relation \eqref{w:comrel} with $g=k_i$ means the weight conservation and it holds due to the factor $\delta_{\alpha + \beta}^{\gamma + \delta}$. In the sequel we show \eqref{w:comrel} for $g=f_i$. The case $g=e_i$ is similar hence omitted. Let the both sides of \eqref{w:comrel} act on $|\alpha\rangle \ot |\beta \rangle \in V_{l}\otimes V_m$ and compare the coef\/f\/icients of $| \gamma\rangle \ot |\delta \rangle \in V_{l}\otimes V_{m}$ in the output vector. Using \eqref{w:Delta}, \eqref{w:rep} and~\eqref{w:matrixele} we f\/ind that the relation to be proved is
\begin{gather*}
 X_{\alpha,\beta}^{\gamma,\delta - \hat{i}} [\delta_{i+1}+1]\theta(\delta_{i+1}=0) +X_{\alpha,\beta}^{\gamma-\hat{i},\delta} q^{\delta_{i}
+\delta_{i+1}}(-1)^{-\delta_{i+1}} z^{-\delta_{i,0}} [\gamma_{i+1}+1]\theta(\gamma_{i+1}=0) \\
\qquad{} = X_{\alpha +\hat{i},\beta}^{\gamma,\delta}
[\alpha_{i+1}]\theta(\alpha_{i+1}=1)z^{-\delta_{i,0}}
+ X_{\alpha,\beta+\hat{i}}^{\gamma,\delta}[\beta_{i+1}]
\theta(\beta_{i+1}=1)q^{\alpha_{i}+\alpha_{i+1}}(-1)^{-\alpha_{i+1}}
\end{gather*}
at $z=q^{l-m}$ under the weight conservation (i) $\alpha_{i}+\beta_{i} = \gamma_{i} + \delta_{i} -1$ and (ii) $\alpha_{i+1}+\beta_{i+1} = \gamma_{i+1} + \delta_{i+1} +1$. By substituting \eqref{w:defQ} and applying Lemma~\ref{w:lempsi}, it is simplif\/ied to
\begin{gather*}
[\delta_{i+1}+1]\big(1-q^{2\beta_{i+1}}\big)\big(1-q^{2(\beta_{i}-\gamma_{i}+1)}\big) \big(1-q^{2(\gamma_{i+1}+1)}\big)\theta(\delta_{i+1}=0) \\
\qquad\quad{} +q^{\gamma_{i+1}+\delta_{i+1} +2\beta_{i}-2\gamma_{i}+2}[\gamma_{i+1}+1] \big(1-q^{2\beta_{i+1}}\big)\big(1-q^{2(\beta_{i+1}-\gamma_{i+1})}\big)\\
\qquad\quad{}\times
\big(1-q^{2\gamma_{i}}\big)(-1)^{\delta_{i+1}}\theta(\gamma_{i+1}=0) \\
\qquad{} =q^{\beta_{i+1}-\gamma_{i+1}} [\alpha_{i+1}]\big(1-q^{2\beta_{i+1}}\big)\big(1-q^{2(\beta_{i}-\gamma_{i}+1)}\big) \big(1-q^{2(\gamma_{i+1}+1)}\big)\theta(\alpha_{i+1}=1) \\
\qquad\quad{} +q^{\gamma_{i+1}+\alpha_{i+1}} [\beta_{i+1}]\big(1-q^{2(\beta_{i}+1)}\big)\big(1-q^{2(\beta_{i+1}-\gamma_{i+1})}\big)
\big(1-q^{2(\gamma_{i+1}+1)}\big)\\
\qquad\quad{}\times (-1)^{\alpha_{i+1}}\theta(\beta_{i+1}=1).
\end{gather*}
From the weight conservation (ii) $\alpha_{i+1}+\beta_{i+1} = \gamma_{i+1} + \delta_{i+1} +1$ and $\ep_{i+1}=1$, we have only four cases
$(\alpha_{i+1},\beta_{i+1},\gamma_{i+1},\delta_{i+1}) =(1,1,1,0),(1,0,0,0),(0,1,0,0),(1,1,0,1)$. They are easily checked by using (i).
\end{proof}

\begin{proof}[Proof of Theorem \ref{th:specialized R}]
The $\mathcal{U}_A$-module $V_{l} \ot V_{m}$ is irreducible for the choice $(\ep_1,\ldots,\ep_{n+1})=(1^{\kappa},0^{n+1-\kappa})$ with $0 \le \kappa \le n+1$
by \cite[Proposition~6.11]{KOS}. (This fact has been quoted as Theorem~\ref{th:true} in this paper.) Therefore $R(z)$ is uniquely characterized by the relation~\eqref{w:comrel} up to normalization. The agreement of the normalization is readily checked.
\end{proof}

If Conjecture \ref{yssi} holds, Proposition \ref{w:proposition2} tells that the factorized formula in Theorem~\ref{th:specialized R} is valid for arbitrary $(\ep_1,\ldots,\ep_{n+1}) \in \{0,1\}^{n+1}$.

\subsection{Parameter version}
For $n \in \Z_{\ge 1}$ and $\epsilon=(\ep_1,\ldots,\ep_{n}) \in \{0,1\}^{n}$, set
\begin{gather*}
B(\epsilon) = \{(\alpha_1,\ldots, \alpha_n) \,|\, \alpha_i \in \{0,1\} \ {\rm if} \ \epsilon_i=1,\
\alpha_i \in \Z_{\ge 0} \ {\rm if} \ \epsilon_i=0\},\\
W(\epsilon) = \bigoplus_{(\alpha_1,\ldots, \alpha_n) \in B(\epsilon)} \C |\alpha_1,\ldots, \alpha_n\rangle.
\end{gather*}
Note that we have shifted to the $n$-component setting. Introduce the operator $\mathscr{S}^{(\epsilon)}(\lambda,\mu) \in \operatorname{End}(W(\epsilon) \otimes W(\epsilon))$ depending on the parameters $\lambda, \mu$ by
\begin{gather}
\mathscr{S}^{(\epsilon)}(\lambda,\mu)(|\alpha\rangle \otimes | \beta\rangle ) = \sum_{\gamma,\delta \in B(\epsilon)}\mathscr{S}(\lambda,\mu)_{\alpha,\beta}^{\gamma,\delta} |\gamma\rangle \otimes | \delta\rangle,\label{smdef2}
\end{gather}
where the element $\mathscr{S}(\lambda,\mu)_{\alpha,\beta}^{\gamma,\delta}$ is specif\/ied by exactly the same formula as~(\ref{lin}) and~(\ref{Pdef}). In other words $\mathscr{S}^{(\epsilon)}(\lambda,\mu)$ is a restriction of $\mathscr{S}(\lambda,\mu) \in \operatorname{End}(W\otimes W)$ on~$W(\epsilon)\ot W(\epsilon)$, where $W$ was def\/ined before~(\ref{smdef}). It corresponds to the parameter version of $R(z=q^{l-m})$ for $(\epsilon_1,\ldots,\epsilon_{n+1})$
with $\epsilon_{n+1}=0$.\footnote{For example $\mathscr{S}^{(1,\ldots,1)}(\lambda, \mu)$ originates in $R^{l,m}(z)$ with `inhomogeneous' choice $(\epsilon_1,\ldots, \epsilon_{n+1})=(1,\ldots,1,0)$.}

Combining the YBE (\ref{yyb}), Theorem \ref{th:specialized R} and the argument similar to \cite[Section~2.3]{KMMO}, one can generalize (\ref{ybe2}) to
\begin{Theorem}
For $\epsilon =(1^\kappa,0^{n-\kappa})$ with $0 \le \kappa \le n$, $\mathscr{S}^{(\epsilon)}(\lambda,\mu)$ satisfies the YBE:
\begin{gather*}
\mathscr{S}^{(\epsilon)}_{1,2}(\nu_1,\nu_2)\mathscr{S}^{(\epsilon)}_{1,3}(\nu_1, \nu_3)\mathscr{S}^{(\epsilon)}_{2,3}(\nu_2, \nu_3)=
\mathscr{S}^{(\epsilon)}_{2,3}(\nu_2, \nu_3)\mathscr{S}^{(\epsilon)}_{1,3}(\nu_1, \nu_3)\mathscr{S}^{(\epsilon)}_{1,2}(\nu_1,\nu_2).
\end{gather*}
\end{Theorem}
In view of Conjecture \ref{yssi} we also conjecture that the above YBE is valid for arbitrary $\epsilon \in \{0,1\}^n$. A direct proof of this assertion will not be dif\/f\/icult although we do not pursue it here.

On the other hand, the sum-to-unity (\ref{sum2}) does not hold in general if $\epsilon \neq (0,\ldots,0)$. Here is the simplest example.

\begin{Example}\label{ykw2}
$\mathscr{S}^{(\epsilon)}(\lambda,\mu)$ with $n=1$ and $\epsilon=(1)$ def\/ines a f\/ive vertex model whose vertex weights read
\begin{gather*}
\includegraphics{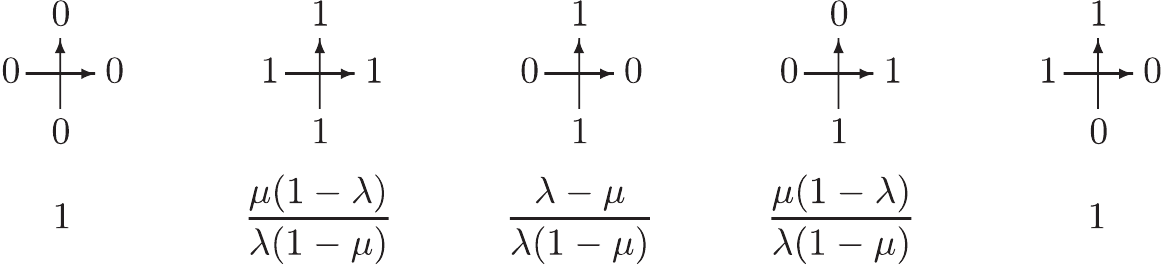}
%
\end{gather*}
in the convention (\ref{vertex}). One sees that the sum-to-unity (\ref{sum2}) is invalid when $(\alpha,\beta)=(1,1)$ because $\mathscr{S}^{(1)}(\lambda,\mu)^{0,2}_{1,1}= \frac{\lambda-\mu}{\lambda(1-\mu)}$ must be dismissed from the sum. These weights also possess the NW-freeness mentioned in Remark~\ref{tsk}.
\end{Example}

\section{Summary}\label{sec7}
We have reviewed the construction of the multispecies ZRP in \cite{KMMO}, the matrix product formula for the stationary probability in \cite{KO1,KO2} and the relevant quantum $R$ matrices originating in the tetrahedron equation and the generalized quantum groups in~\cite{KOS}. We have also pointed out a new commuting Markov transfer matrix in Section~\ref{sec2.3}, the associated Markov Hamiltonian~(\ref{ri3}), the Serre type relations (\ref{w:relA4})--(\ref{w:relA7}) for the generalized quantum group $\U_A(\epsilon_1,\ldots, \epsilon_{n+1})$, factorization of its quantum $R$ matrix at special point of the spectral parameter in Theorem~\ref{th:specialized R}, and the parameter version of the $R$ matrix~(\ref{smdef2}) satisfying the YBE.

\appendix
\section[Example of $R^{l,m}(z)$ for $\mathcal{U}_A(1,1,0)$]{Example of $\boldsymbol{R^{l,m}(z)}$ for $\boldsymbol{\mathcal{U}_A(1,1,0)}$}\label{sec.app} \allowdisplaybreaks

\looseness=1 Consider $\mathcal{U}_{A}(1,1,0)$. For $l,m \ge 2$, one has $V_{m} = \C|0,0,m\rangle \oplus \C|0,1,m-1 \rangle \oplus \C|1,0,m-1 \rangle \oplus \C|1,1,m-2 \rangle \subset \mathcal{W} = V \ot V \ot F$ and similarly for $V_{l}$. The action of $R(z)$ on $V_{l} \ot V_{m}$ is given by
\begin{gather*}
R(z)(|0,0,l \rangle \ot |0,0,m \rangle) = |0,0,l \rangle \ot |0,0,m \rangle, \\
R(z)(|0,0,l \rangle \ot |0,1,m-1 \rangle) =\frac{q^{l}z-q^{m}}{z-q^{l+m}}|0,0,l \rangle \ot |0,1,m-1 \rangle\\
\hphantom{R(z)(|0,0,l \rangle \ot |0,1,m-1 \rangle) =}{} + \frac{(1-q^{2l})z}{z-q^{l+m}} |0,1,l-1 \rangle \ot |0,0,m \rangle, \\
R(z)(|0,0,l \rangle \ot |1,0,m-1 \rangle) =\frac{q^{l}z-q^{m}}{z-q^{l+m}}|0,0,l \rangle \ot |1,0,m-1 \rangle\\
\hphantom{R(z)(|0,0,l \rangle \ot |0,1,m-1 \rangle) =}{} + \frac{(1-q^{2l})z}{z-q^{l+m}} |1,0,l-1 \rangle \ot |0,0,m \rangle,
\\
R(z)(|0,0,l \rangle \ot |1,1,m-2 \rangle) =\frac{(q^{m}-q^{l}z)(q^{m}-q^{l+2}z)}{(q^{l+m}-z)(q^{l+m}-q^{2}z)}
|0,0,l \rangle \ot |1,1,m-2 \rangle \\
\hphantom{R(z)(|0,0,l \rangle \ot |1,1,m-2 \rangle) =}{} + \frac{(q^{l}z-q^{m})(1-q^{2l})zq^{2}}{(q^{l+m}-z)(q^{l+m}-q^{2}z)}
|0,1,l-1 \rangle \ot |1,0,m-1 \rangle \\
\hphantom{R(z)(|0,0,l \rangle \ot |1,1,m-2 \rangle) =}{}+ \frac{(q^{l}z-q^{m})(1-q^{2l})zq}{(q^{l+m}-z)(q^{l+m}-q^{2}z)}
|1,0,l-1 \rangle \ot |0,1,m-1 \rangle \\
\hphantom{R(z)(|0,0,l \rangle \ot |1,1,m-2 \rangle) =}{}+ \frac{(1-q^{2l})(1-q^{2l-2})z^{2}q^{2}}{(q^{l+m}-z)(q^{l+m}-q^{2}z)} |1,1,l-2 \rangle \ot |0,0,m \rangle,
\\
R(z)(|1,1,l-2 \rangle \ot |0,0,m \rangle)=\frac{q^{2}(1-q^{2m})(1-q^{2m-2})}{(q^{l+m}-z)(q^{l+m}-q^{2}z)}|0,0,l \rangle \ot |1,1,m-2 \rangle \\
\hphantom{R(z)(|1,1,l-2 \rangle \ot |0,0,m \rangle)=}{} + \frac{q^{2}(1-q^{2m})(q^{m}z-q^{l})}{(q^{l+m}-z)(q^{l+m}-q^{2}z)}
|0,1,l-1 \rangle \ot |1,0,m-1 \rangle \\
\hphantom{R(z)(|1,1,l-2 \rangle \ot |0,0,m \rangle)=}{}+ \frac{q(1-q^{2m})(q^{m}z-q^{l})}{(q^{l+m}-z)(q^{l+m}-q^{2}z)}
|1,0,l-1 \rangle \ot |0,1,m-1 \rangle \\
\hphantom{R(z)(|1,1,l-2 \rangle \ot |0,0,m \rangle)=}{}+ \frac{(q^{m}z-q^{l})(q^{m+2}z-q^{l})}{(q^{l+m}-z)(q^{l+m}-q^{2}z)}
|1,1,l-2 \rangle \ot |0,0,m \rangle,\\
R(z)(|1,1,l-2 \rangle \ot |0,1,m-1 \rangle)=\frac{(1-q^{l+m}z)(1-q^{2m-2})}{(q^{l+m}-z)(q^{l+m-2}-z)}
|0,1,l-1 \rangle \ot |1,1,m-2 \rangle \\
\hphantom{R(z)(|1,1,l-2 \rangle \ot |0,1,m-1 \rangle)=}{} + \frac{(1-q^{l+m}z)(q^{m-1}z-q^{l-1})}{(q^{l+m}-z)(q^{l+m-2}-z)}|1,1,l-2 \rangle \ot |0,1,m-1 \rangle,\!
\\
R(z)(|1,1,l-2 \rangle \ot |1,0,m-1 \rangle)=\frac{(1-q^{l+m}z)(1-q^{2m-2})}{(q^{l+m}-z)(q^{l+m-2}-z)}|1,0,l-1 \rangle \ot |1,1,m-2 \rangle \\
\hphantom{R(z)(|1,1,l-2 \rangle \ot |1,0,m-1 \rangle)=}{}
+ \frac{(1-q^{l+m}z)(q^{m-1}z-q^{l-1})}{(q^{l+m}-z)(q^{l+m-2}-z)} |1,1,l-2 \rangle \ot |1,0,m-1 \rangle,\!
\\
R(z)(|1,1,l-2 \rangle \ot |1,1,m-2 \rangle) =\frac{(1-q^{l+m}z)(1-q^{l+m-2}z)}{(q^{l+m}-z)(q^{l+m-2}-z)}|1,1,l-2 \rangle \ot |1,1,m-2 \rangle, \\
R(z)(|1,0,l-1 \rangle \ot |0,0,m \rangle)=\frac{1-q^{2m}}{z-q^{l+m}}|0,0,l \rangle \ot |1,0,m-1 \rangle\\
\hphantom{R(z)(|1,0,l-1 \rangle \ot |0,0,m \rangle)=}{}
+\frac{q^{m}z-q^{l}}{z-q^{l+m}}|1,0,l-1 \rangle \ot |0,0,m \rangle,\\
R(z)(|1,0,l-1 \rangle \ot |0,1,m-1 \rangle)=\frac{(q^{2m}-q^{2})(q^{m}-q^{l}z)}{(q^{l+m}-z)(q^{l+m}-q^{2}z)}|0,0,l \rangle \ot |1,1,m-2 \rangle \\
\hphantom{R(z)(|1,0,l-1 \rangle \ot |0,1,m-1 \rangle)=}{}
+ \frac{(q^{2}-1)q^{l+m} + (q^{2}-q^{2+2l}-q^{2+2m}+q^{2l+2m})z}{(q^{l+m}-z)(q^{l+m}-q^{2}z)}\\
\hphantom{R(z)(|1,0,l-1 \rangle \ot |0,1,m-1 \rangle)=}{}\times
|0,1,l-1 \rangle \ot |1,0,m-1 \rangle \\
\hphantom{R(z)(|1,0,l-1 \rangle \ot |0,1,m-1 \rangle)=}{} + \frac{q(q^{m}-q^{l}z)(q^{l}-q^{m}z)}{(q^{l+m}-z)(q^{l+m}-q^{2}z)}
|1,0,l-1 \rangle \ot |0,1,m-1 \rangle \\
\hphantom{R(z)(|1,0,l-1 \rangle \ot |0,1,m-1 \rangle)=}{} + \frac{(q^{2l}-q^{2})(q^{l}-q^{m}z)z}{(q^{l+m}-z)(q^{l+m}-q^{2}z)}
|1,1,l-2 \rangle \ot |0,0,m \rangle,
\\
R(z)(|1,0,l-1 \rangle \ot |1,0,m-1 \rangle)=\frac{1-q^{l+m}z}{z-q^{l+m}} |1,0,l-1 \rangle \ot |1,0,m-1 \rangle, \\
R(z)(|1,0,l-1 \rangle \ot |1,1,m-2 \rangle)=\frac{q(1-q^{l+m}z)(q^{l}z-q^{m})}{(q^{l+m}-z)(q^{l+m}-q^{2}z)}
|1,0,l-1 \rangle \ot |1,1,m-2 \rangle \\
\hphantom{R(z)(|1,0,l-1 \rangle \ot |1,1,m-2 \rangle)=}{}
+\frac{q^{2}z(1-q^{l+m}z)(1-q^{2l-2})}{(q^{l+m}-z)(q^{l+m}-q^{2}z)} |1,1,l-2 \rangle \ot |1,0,m-1 \rangle, \\
R(z)(|0,1,l-1 \rangle \ot |0,0,m \rangle)=\frac{1-q^{2m}}{z-q^{l+m}} |0,0,l \rangle \ot |0,1,m-1 \rangle\\
\hphantom{R(z)(|0,1,l-1 \rangle \ot |0,0,m \rangle)=}{}
+\frac{q^{m}z-q^{l}}{z-q^{l+m}}|0,1,l-1 \rangle \ot |0,0,m \rangle, \\
R(z)(|0,1,l-1 \rangle \ot |0,1,m-1 \rangle) =\frac{1-q^{l+m}z}{z-q^{l+m}}|0,1,l-1 \rangle \ot |0,1,m-1 \rangle,
\\
R(z)(|0,1,l-1 \rangle \ot |1,0,m-1 \rangle) =\frac{q(q^{l}z-q^{m})(1-q^{2m-2})}{(q^{l+m}-z)(q^{l+m}-q^{2}z)}
|0,0,l \rangle \ot |1,1,m-2 \rangle \\
\hphantom{R(z)(|0,1,l-1 \rangle \ot |1,0,m-1 \rangle) =}{}
+ \frac{q(q^{l}z-q^{m})(q^{m}z-q^{l})}{(q^{l+m}-z)(q^{l+m}-q^{2}z)}|0,1,l-1 \rangle \ot |1,0,m-1 \rangle \\
\hphantom{R(z)(|0,1,l-1 \rangle \ot |1,0,m-1 \rangle) =}{}
+\frac{q^{l+m}(1-q^{2})z^{2} +(q^{2}+q^{2l+2m}-q^{2l}-q^{2m})z}{(q^{l+m}-z)(q^{l+m}-q^{2}z)}\\
\hphantom{R(z)(|0,1,l-1 \rangle \ot |1,0,m-1 \rangle) =}{}\times
|1,0,l-1 \rangle \ot |0,1,m-1 \rangle \\
\hphantom{R(z)(|0,1,l-1 \rangle \ot |1,0,m-1 \rangle) =}{}
+ \frac{qz(1-q^{2l-2})(q^{m}z-q^{l})}{(q^{l+m}-z)(q^{l+m}-q^{2}z)} |1,1,l-2 \rangle \ot |0,0,m \rangle,
\\
R(z)(|0,1,l-1 \rangle \ot |1,1,m-2 \rangle)=\frac{(1-q^{l+m}z)(q^{l-1}z-q^{m-1})}{(z-q^{l+m})(z-q^{l+m-2})}
|0,1,l-1 \rangle \ot |1,1,m-2 \rangle \\
\hphantom{R(z)(|0,1,l-1 \rangle \ot |1,1,m-2 \rangle)=}{}
+\frac{(1-q^{l+m}z)(1-q^{2l-2})z}{(z-q^{l+m})(z-q^{l+m-2})} |1,1,l-2 \rangle \ot |0,1,m-1 \rangle .
\end{gather*}

\subsection*{Acknowledgments}

The authors thank Ivan Corwin, Philippe Di Francesco, Alexandr Garbali, Michio Jimbo and Tomohiro Sasamoto for kind interest. They also thank Jef\/frey Kuan for informing them of the interesting work~\cite{Kuan} and Shohei Machida for letting them know the Serre relations of~$\U_A(\epsilon)$. Last but not least we thank the anonymous referees for productive suggestions to improve the paper. This work is supported by Grants-in-Aid for Scientif\/ic Research No.~15K04892, No.~15K13429 and No.~16H03922 from JSPS.

\newpage

\pdfbookmark[1]{References}{ref}
\LastPageEnding

\end{document}